\def\mode{0} 

\if0\mode
\documentclass[journal, twoside, web]{ieeecolor}
\usepackage{generic}
\else
\documentclass[onecolumn, draftclsnofoot]{IEEEtran}
\fi

\usepackage[english]{babel}
\usepackage[utf8]{inputenc}
\usepackage[T1]{fontenc}
\usepackage{microtype}
\usepackage{setspace}
\usepackage{blindtext}
\usepackage{siunitx}
\usepackage{comment}
\usepackage{xargs}
\usepackage{comment}
\usepackage{cancel}
\usepackage{titlecaps}
\Addlcwords{or if of and an to via -off on off for in the vs} 

\usepackage[justification=raggedright,singlelinecheck=false,size=footnotesize]{caption}
\usepackage[justification=raggedright,singlelinecheck=false,size=footnotesize]{subcaption}
\usepackage{fancyhdr} 
\usepackage{titlecaps}
\Addlcwords{or the if a an of for and to -off off in not by via}
\usepackage[pdftex]{graphicx}
\usepackage{color}
\usepackage{xcolor}
\usepackage{epsfig}
\usepackage{epstopdf}
\graphicspath{{figures/}}
\usepackage{tikz}
\usepackage{pgfplots,pgfplotstable}
\usetikzlibrary{patterns,fit,matrix}
\usepackage{tikzscale}
\usepackage{scalerel}
\usepgfplotslibrary{groupplots,dateplot}
\pgfplotsset{compat=newest}
\pgfplotsset{every axis/.append style={
		label style={font=\normalsize},
		tick label style={font=\normalsize}  
}}
\DeclareGraphicsExtensions{.pdf,.jpeg,.png,.jpg}
\usepackage{tikz}
\usepackage{tikzscale}
\usepackage{scalerel}
\usepackage{tikz-layers}
\usetikzlibrary{patterns,
	arrows,
	plotmarks,
	shapes.arrows,
	spy,
	svg.path,
	shapes.multipart,
	positioning,
	decorations,
	calc,
	decorations.markings,
	matrix,
}
\tikzset{add reference/.style={insert path={%
			coordinate [pos=0,xshift=-0.5\pgflinewidth,yshift=-0.5\pgflinewidth] (#1 south west) 
			coordinate [pos=1,xshift=0.5\pgflinewidth,yshift=0.5\pgflinewidth]   (#1 north east)
			coordinate [pos=.5] (#1 center)                        
			(#1 south west |- #1 north east)     coordinate (#1 north west)
			(#1 center     |- #1 north east)     coordinate (#1 north)
			(#1 center     |- #1 south west)     coordinate (#1 south)
			(#1 south west -| #1 north east)     coordinate (#1 south east)
			(#1 center     -| #1 south west)     coordinate (#1 west)
			(#1 center     -| #1 north east)     coordinate (#1 east)   
}}}
\usepackage{booktabs}

\usepackage{array}
\usepackage{stfloats}
\usepackage{algorithmic}
\usepackage{textcomp}
\let\labelindent\relax
\usepackage{enumitem}
\usepackage{booktabs}
\usepackage{multirow}
\usepackage[absolute]{textpos}

\usepackage{amsmath,amssymb,amsfonts,amsthm}
\usepackage{mathdots}
\usepackage{bbm,dsfont}
\usepackage{relsize}
\usepackage{nicefrac}
\usepackage{bbm,bm}
\usepackage{mathtools}
\usepackage[mathscr]{eucal}
\usepackage{optidef}
\usepackage[nomessages]{fp}

\usepackage[nocompress]{cite}

\usepackage{url}
\usepackage[colorlinks,citecolor=blue,linkcolor=blue]{hyperref}
\usepackage[capitalize,nameinlink]{cleveref}
\makeatletter
\newcommand*{\Crefns}[1]{{\@cref@sortfalse\@cref@compressfalse\Cref{#1}}}
\makeatother
\usepackage{orcidlink}

\if0\mode
\newtheoremstyle{colorthm}{\topsep}{\topsep}{\normalfont}{1em}{\color{nblue}\sf\itshape\small\bfseries}{:}{ }{\thmname{#1}\thmnumber{ #2}{\thmnote{ (#3)}}}
\theoremstyle{colorthm}
\fi

\newtheorem{thm}{Theorem}

\newtheorem{prop}{Proposition}

\newtheorem{definition}{Definition}
\newtheorem{prob}{Problem}
\newtheorem{ass}{Assumption}

\newtheorem{ex}{Example}

\Crefname{cor}{Corollary}{Corollaries}
\Crefname{prop}{Proposition}{Propositions}
\Crefname{conj}{Conjecture}{Conjectures}
\Crefname{ass}{Assumption}{Assumptions}
\Crefname{figure}{Figure}{Figs.}
\Crefname{ex}{Example}{Examples}
\Crefname{prob}{Problem}{Problems}

\newcommand{\calD}{{\cal D}}
\newcommand{\calE}{{\cal E}}

\newcommand{\calG}{{\cal G}}
\newcommand{\calH}{{\cal H}}
\newcommand{\calI}{{\cal I}}

\newcommand{\calN}{{\cal N}}
\newcommand{\calO}{{\cal O}}
\newcommand{\calP}{{\cal P}}

\newcommand{\calV}{{\cal V}}

\newcommand{\lr}{\left(}
\newcommand{\rr}{\right)}

\newcommand{\lb}{\left\lbrace}
\newcommand{\rb}{\right\rbrace}

\newcommand{\Real}[1]{ { {\mathbb R}^{#1} } }

\DeclareMathOperator*{\argmax}{arg\;max}
\DeclareMathOperator*{\argmin}{arg\;min}

\DeclareMathOperator*{\softmax}{softmax}

\newcommand{\DeclareAutoPairedDelimiter}[3]{%
	\expandafter\DeclarePairedDelimiter\csname Auto\string#1\endcsname{#2}{#3}%
	\begingroup\edef\x{\endgroup
		\noexpand\DeclareRobustCommand{\noexpand#1}{%
			\expandafter\noexpand\csname Auto\string#1\endcsname*}}%
	\x}

\DeclareAutoPairedDelimiter{\ceil}{\lceil}{\rceil}
\DeclareAutoPairedDelimiter{\floor}{\lfloor}{\rfloor}
\newcommand{\norm}[2][]{\left\lVert#2\right\rVert^2_{#1}}

\renewcommand{\b}[1]{\boldsymbol{#1}}

\newcommand{\Sum}[2]{#1 + #2 = \the\numexpr #1 + #2 \relax \\}

\makeatletter
\newcommand{\data}[1][]{\calD_{#1}}
\newcommand{\datapoint}[2]{d%
	\if#1\@empty \else  _{#1,#2} \fi}
\newcommand{\feat}[2]{a%
	\if#1\@empty \else  _{#1,#2} \fi}
\newcommand{\lab}[2]{b%
	\if#1\@empty \else  _{#1,#2} \fi}
\newcommand{\idxset}[1]{\calI_{#1}}

\makeatother

\addto\captionsenglish{\renewcommand{\figurename}{Fig.}}
\addto\captionsenglish{\renewcommand{\tablename}{TABLE}}

\newcommand{\blue}[1]{#1}

\newcommand{\linkToPdf}[1]{\href{#1}{\blue{(pdf)}}}
\newcommand{\linkToPpt}[1]{\href{#1}{\blue{(ppt)}}}
\newcommand{\linkToCode}[1]{\href{#1}{\blue{(code)}}}
\newcommand{\linkToWeb}[1]{\href{#1}{\blue{(web)}}}
\newcommand{\linkToVideo}[1]{\href{#1}{\blue{(video)}}}
\newcommand{\linkToMedia}[1]{\href{#1}{\blue{(media)}}}
\newcommand{\award}[1]{\xspace} 
\newcommand{\eg}{\emph{e.g.,}\xspace}
\newcommand{\ie}{\emph{i.e.,}\xspace}
\if0\mode
\addto\extrasenglish{}
\addto\extrasenglish{}
\else
\addto\extrasenglish{}
\addto\extrasenglish{}
\fi

\addto\extrasenglish{}
\addto\extrasenglish{}
\addto\extrasenglish{}

\title{Personalized and Resilient Distributed Learning Through Opinion Dynamics}

\author{Luca~Ballotta\textsuperscript{\orcidlink{0000-0002-6521-7142}},~\IEEEmembership{Member,~IEEE}, Nicola~Bastianello\textsuperscript{\orcidlink{0000-0002-5634-8802}},~\IEEEmembership{Member,~IEEE}, Riccardo~M.~G.~Ferrari\textsuperscript{\orcidlink{0000-0003-3615-5445}},~\IEEEmembership{Senior~Member,~IEEE}, and~Karl~H.~Johansson\textsuperscript{\orcidlink{0000-0001-9940-5929}},~\IEEEmembership{Fellow,~IEEE}
	\thanks{This work was supported in part
		by the EU Horizon-RIA Projects TWAIN under Grant 101122194 and ULTIMATE under Grant 101070162,
        and in part by the Swedish Research Council Distinguished Professor Grant 2017-01078 Knut and Alice Wallenberg Foundation Wallenberg Scholar Grant.
		\textit{(The first two authors contributed equally to this work).}}%
	\thanks{Luca Ballotta is with the Department of Information Engineering, University of Padova, 35131 Padova, Italy
		(e-mail: luca.ballotta@unipd.it).
		Work done while affiliated with the Delft University of Technology.}
	\thanks{Riccardo M. G. Ferrari is with the Delft Center for Systems and Control (DCSC), Delft University of Technology, 2628 CD Delft, Netherlands
		(e-mail: r.ferrari@tudelft.nl).}%
	\thanks{Nicola Bastianello and Karl H. Johansson are with the School of Electrical Engineering and Computer Science, and Digital Futures, KTH Royal Institute of Technology, Sweden.
		(e-mail: \{nicolba, kallej\}@kth.se).}%
	}

\begin{document}
	
	\if0\mode
	\begin{textblock}{20}(-2,0.05)
		\footnotesize
		\centering
		\setstretch{1}
		This article has been accepted for publication in the IEEE Transactions on Control of Network Systems.
		Please cite the paper as:\\
		L. Ballotta, N. Bastianello, R. M. G. Ferrari, and K. H. Johansson,\\
		``\titlecap{Personalized and Resilient Distributed Learning Through Opinion Dynamics},''\\
		IEEE Transactions on Control of Network Systems, 2026.
	\end{textblock}
	\fi
	
 	\bstctlcite{bib-options}
	
	\maketitle
	

\begin{abstract}
	
	In this paper,
	we address two practical challenges of distributed learning in multi-agent network systems,
	namely personalization and resilience.
	Personalization is the need of heterogeneous agents to learn local models tailored to their own data and tasks, while still generalizing well;
	on the other hand,
	the learning process must be resilient to cyberattacks or anomalous training data to avoid disruption.
	Motivated by a conceptual affinity between these two requirements,
	we devise a distributed learning algorithm that combines distributed gradient descent and the Friedkin-Johnsen model of opinion dynamics to fulfill both of them.
	We quantify its convergence speed and the neighborhood that contains the final learned models,
	which can be easily controlled by tuning the algorithm parameters to enforce a more personalized/resilient behavior.
	We numerically showcase the effectiveness of our algorithm on synthetic and real-world distributed learning tasks,
	where it achieves high global accuracy both for personalized models and with malicious agents compared to standard strategies.
		
\end{abstract}

\section{Introduction}\label{sec:intro}

\IEEEPARstart{D}{istributed} learning is the natural extension of machine learning to multi-agent systems where devices collaboratively train models.
Applications include multi-robot systems~\cite{Yu22ral-DiNNO},
connected vehicles~\cite{Ma21tvt-distributedLearningPlatoon}, 
smart grids~\cite{Huang18acm-distributedLearningSmartGrid}, 
wind power forecasting~\cite{Sommer21ijf-distributedLearningWind},
and the Internet-of-Things~\cite{Le24cst-distributedLearningIoT}.
Distributed learning is more robust than federated learning whereby the central aggregator may bias the global model towards some agents and is vulnerable to cyberattacks~\cite{Ma23pieee-TrustedAIinMultiagentSystems}.
Yet, other problems are present.
In this paper,
we focus on two key requirements in multi-agent network systems:
1) personalization reflects the need of heterogeneous agents to prioritize their own local data and tasks;
2) resilience is the capability of agents to deter unknown intruders trying to disrupt the learning process.

\subsection{Literature Review}\label{sec:related-works}

Distributed learning is based on recent developments in distributed optimization~\cite{Nedic18pieee-distributedOptimizationSurvey,Notarstefano19fdsc-distributedOptimizationSurvey}.
Among distributed optimization algorithms,
distributed gradient descent (DGD), gradient tracking (GT), and dual methods (especially ADMM) are widespread~\cite{Notarstefano19fdsc-distributedOptimizationSurvey}.
Algorithms based on GT and ADMM achieve exact convergence to an optimal solution using fixed parameters, 
while DGD converges inexactly unless a vanishing step-size,
which slows down convergence,
is employed.
These algorithms are leveraged to solve distributed learning problems~\cite{Xin20spm-DecentralizedStochasticOptimization,Chang20spm-DistributedLearningNonconvex}.

Most works assume that all agents are honest. 
However, this may not be the case as an adversary can deploy attacks through the wireless network~\cite{Ma23pieee-TrustedAIinMultiagentSystems,Lyu24tnnls-PrivacyRobustnessFL,Liu24arxiv-distributedLearningSurvey}. 
Attacks can target either data, to infer or corrupt the agents' datasets,
or models,
to worsen their accuracy or bias.
In this paper we focus on model attacks,
which requires securing the training algorithm.
A growing literature is devoted to resilient distributed algorithms,
see~\cite{Liu24arxiv-distributedLearningSurvey} for a survey.
A core component of distributed algorithms is consensus averaging, 
during which each agent exchanges models and possibly gradients with neighbors and updates its own local model based on received information.
This step is vulnerable because attackers can transmit malicious data and pollute the agents' updates.
While actively discriminating malicious agents to disregard their model updates is effective,
it requires cybersecurity mechanisms not suited for low-power devices with limited hardware or fast algorithm execution.
Hence,
designing a passively resilient distributed algorithm with low computational footprint is of utmost importance.
To this aim,
resilient distributed algorithms replace the consensus step with a robust protocol.
A common strategy uses trimmed means without the largest and smallest received values~\cite{Sundaram19tac-resilientDistributedOptimization,Fang22tsipn-BRIDGE},
which requires dense inter-agent connectivity and a shared bound on the number of adversaries.
Shang~\cite{Shang22tcs-medianBasedResilientConsensus} replaces the average with the median,
but still requires dense connectivity.
A related algorithm in~\cite{Li22tro-resilientDistributedLearningRobots} uses the centerpoint,
a generalization of the median to higher dimensions,
but assumes that all agents have the same optimizer.
Abbas \emph{et al.}~\cite{Abbas18tcns-trustedNodes} leverage  ``trusted'' agents which must be secured from cyberattacks.
A line of works uses ``trust'' obtained from physical information channels to filter out malicious messages,
whose performance depends on the statistics of trust observations~\cite{Yemini25tac-resilientDistributedOptim,Ballotta24acc-trustConfidence,Hadjicostis22cdc-trustworthyConsensus}.
\blue{By contrast,
in this work we focus effort on a solution that does not require dense connectivity among agents,
shared parameters,
or extra resources and provides guarantees under standard assumptions.}

Another challenge of distributed learning is the statistical heterogeneity of the cooperating agents~\cite{Fallah20nips-personalizedFL}. 
The agents' local data are often drawn from different distributions which reflect their peculiar perspectives (\textit{e.g.}, different information sources, sensors, or plant specifications).
As a consequence,
each agent may wish to learn a model which is highly accurate on its own local distribution,
possibly trading this need with lower accuracy on other agents' distributions.
This concept is referred to as \emph{personalization}.
Personalized learning has been widely explored in federated learning wherein it relies on the central coordinator~\cite{Fallah20nips-personalizedFL,Dinh20nips-personalizedFL,Hanzely22tmlr-personalizedFL}.
Also,
it has been sparsely addressed in distributed learning~\cite{Sadiev22jco-decentralizedPersonalizedFL,Toghani23lcss-PARSPush},
but these approaches transform the original problem to a computationally burdensome bi-level optimization,
especially challenging with compute-intensive models such as in deep learning.
On the contrary,
here we are interested in embedding personalization into the algorithm itself rather than the problem formulation to avoid an excessive computational burden.

The synergy of personalization with other objectives has been explored in federated learning.
Kundu \emph{et al.}~\cite{Kundu22edge-robustnessPersonalizationFL} integrate personalization with robustness to outliers, that is, agents with a significantly different local distribution.
Han \emph{et al.}~\cite{Han23infocom-SplitGP} discuss how personalization and generalization can both be achieved via tailored algorithm design.
Bietti \emph{et al.}~\cite{Bietti22icml-personalizationPrivacyAccuracy} explore the impact of personalized learning for privacy preservation (\ie robustness against attacks on data).
Li \emph{et al.}~\cite{Li21pmlr-Ditto} discuss the potential of personalization to achieve robustness to model attacks and fairness, which ensures uniform performance of the trained models across agents.
However,
all these works do not apply to a fully distributed setup and rely on computationally intensive reformulations of the problem to achieve personalization.

\subsection{Contribution}\label{sec:contribution}

We propose a novel distributed learning algorithm to achieve personalization and resilience.
The idea is to purposely bias,
in a controlled way, 
each agent towards its own local cost.
In collaborative but heterogeneous settings,
this strategy improves local accuracy achieved by the agents,
enhancing personalization;
with malicious agents,
it makes training resilient by reducing influence of  those on the learned models.
In light of this connection,
we develop one algorithm to tackle both problems.
We draw inspiration from the resilient consensus approach in~\cite{Ballotta24tac-competitionCollaboration} and combine the Friedkin-Johnsen (FJ) opinion dynamics model,
originally conceived to capture disagreement,
with DGD.
The key ingredient of our algorithm is a scalar parameter $\lambda\in[0,1]$,
modeling opinion ``stubbornness'' in the original FJ model,
that allows agents to smoothly transition from collaborative training ($\lambda=0$),
which targets high global accuracy irrespectively of heterogeneous agents or attacks,
to local training ($\lambda=1$),
which achieves high local accuracy but cannot generalize well.
We characterize the geometric convergence rate of our algorithm and the distance between learned models and optimum of the nominal distributed learning problem as functions of design parameters.
We conduct an extensive numerical campaign with synthetic and real-world classifications problems.
Our algorithm achieves superior personalization compared to DGD and enhances resilience with over $10\%$ of malicious agents scattered across a sparse network,
improving accuracy of DGD by up to $78\%$.
To the best of our knowledge,
this is the first approach for personalized distributed learning that does not require a computationally intensive reformulation of the original problem.
Further,
it does not require dense communication or extra resources such as secured agents or trust information,
and enjoys formal guarantees under standard assumptions on the cost functions.
In the context of resilience,
our algorithm may be used as a first defense mechanism to let agents partially collaborate from the start till the adversaries are detected by more sophisticated but slower strategies,
such as standard network security or the trust-based algorithms in~\cite{Ballotta24acc-trustConfidence,Hadjicostis22cdc-trustworthyConsensus,Yemini25tac-resilientDistributedOptim}.

\subsubsection*{Organization}
We introduce the distributed learning setup in \autoref{sec:setup},
describing in detail the concepts of personalization (\autoref{sec:what-personalization}) and resilience (\autoref{sec:what-resilience}),
and discussing their affinity (\autoref{sec:affinity}).
In \autoref{sec:algorithm-design} we develop our distributed learning algorithm by combining DGD and the FJ model.
In \autoref{sec:convergence-analysis} we analyze its fixed point and convergence speed in both cases with fully collaborative agents and with malicious agents sending bounded values.
In \autoref{sec:experiments} we test our algorithm on synthetic and real-world classification tasks,
where it outperforms DGD and local training in achieving personalization and resilience.
We discuss current limitations and potential directions of improvements in \autoref{sec:discussion},
and draw conclusions in \autoref{sec:conclusion}.

\section{Distributed Learning Setup}\label{sec:setup}

Consider $N$ agents that exchange information over a wireless network modeled as an undirected graph $\calG = (\calV,\calE)$, where $\calV \doteq \{1,\dots,N\}$ and $(i,j)\in\calE$ means that there exists a direct communication link between the two agents labeled $i$ and $j$.
Agent $i$ is equipped with a decision variable $x_i\in\Real{n}$ and cost function $f_i : \Real{n} \to \Real{}$.
\blue{Notation $\{x_i\}_{i\in\calV}$ denotes the set of all agents' variables.}
Each agent directly manipulates only its variable $x_i$ but can both access its cost function $f_i$ and exchange information with neighbors.
In distributed optimization,
the agents cooperatively solve the following problem
\begin{mini}
	{\{x_i\}_{i\in\calV}}
	{\sum_{i\in\calV} f_i(x_i)}
	{\label{eq:distr-learning}}
	{}
    \addConstraint{x_i}{=x_j}{\ \forall i,j\in\calV.}
\end{mini}

We turn to the specific case of distributed learning,
which is the focus of this paper.
We first recap the underlying probabilistic optimization framework.
Here,
agent $i$ samples data from a local probability distribution $\mathcal{P}_i$ that reflects the agent's perspective on,
or accessed portion of,
the global phenomenon being observed by the network.
Each agent $i$ aims to learn a model,
parameterized by $x_i$,
able to describe (or ``explain'') both its local distribution $\calP_i$ and those of the other agents $\{\calP_j\}_{j\neq i}$.\footnote{
    This is effectively possible only if the distributions $\{\calP_i\}_{i\in\calV}$ are ``similar'' enough.
}
Hence,
the agents cooperate to solve the risk-minimization problem
\begin{mini}
	{\{x_i\}_{i\in\calV}}
	{\sum_{i\in\calV} \mathbb{E}_{d \sim \mathcal{P}_i} \left[ \ell(x_i; d) \right]}
	{\label{eq:risk-min}}
	{}
    \addConstraint{x_i}{=x_j}{\ \forall i,j\in\calV.}
\end{mini}
In problem~\eqref{eq:risk-min},
loss function $\ell : \Real{n} \times \Real{p} \to \Real{}$ quantifies how accurately the model $x_i$ describes a data point $d$ sampled from $\mathcal{P}_i$.
A classic example is image classification
where the data are image-class pairs $d=(\feat{}{},\lab{}{})$ and the model matches an image $\feat{}{}$ to a semantic class $\lab{}{}$ such as an alphanumeric character or an object.
In this case, $x_i$ minimizes
$\ell(x_i;(\feat{}{},\lab{}{}))$ if the model outputs $\lab{}{}$ given input $\feat{}{}$.
However,
problem~\eqref{eq:risk-min} can be solved only if all distributions $\{\mathcal{P}_i\}_{i\in\calV}$ are known, 
which is not the case in practice.
Thus,
the agents address a deterministic approximation of~\eqref{eq:distr-learning} called empirical risk minimization and characterized as follows.
Agent $i$ owns a private dataset $\data[i] =\{\datapoint{i}{j}\}_{j\in\idxset{i}}$ where each data point is assumed sampled from the local distribution,
namely $\datapoint{i}{j} \sim \calP_i$.
This allows the agent $i$ to approximate its risk function $\mathbb{E}_{d \sim \mathcal{P}_i} \left[ \ell(x_i; d) \right]$ with the empirical risk associated with $\data[i]$ and to seek a learning-based model parameterized by $x_i$ that suitably describes the data in $\data[i]$.\footnote{
	Since we assume that all agents learn the same model (\eg a linear model or a neural network model) and only differ in the model parameters $x_i$,
	in the following we refer to $x_i$ just as ``parameter(s)'' for the sake of simplicity.
}
The local cost $f_i$ is consequently instantiated as
\begin{equation}\label{eq:loss}
	f_i(x_i) = \frac{1}{|\idxset{i}|} \sum_{j\in\idxset{i}} \ell(x_i; \datapoint{i}{j}) + \gamma\rho(x_i).
\end{equation}
The regularization function $\rho$ in~\eqref{eq:loss} is designed to prevent overfitting the training dataset $\data[i]$ and usually chosen as a norm with weight $\gamma>0$.
In words,
minimization of the (regularized) empirical risk~\eqref{eq:loss} aims to make agent $i$ learn a parameter $x_i$ that describes $\data[i]$ and also generalizes to the distribution $\calP_i$,
namely it explains samples $d\sim \calP_i$ not present in $\data[i]$.

Problem~\eqref{eq:distr-learning} can be solved by training one model on the aggregated dataset $\{\data[i]\}_{i\in\calV}$.
However,
privacy concerns and communication constraints prevent the agents from sharing data, ruling out this option.
At the same time,
if each agent trains a model only on its own dataset,
it achieves poor generalization.
These concurrent issues make distributed optimization algorithms suited to the distributed-learning specialization of~\eqref{eq:distr-learning} where the local cost functions $f_i$ are the regularized losses defined in~\eqref{eq:loss}.
The agents solve~\eqref{eq:distr-learning} without exchanging the data $\data[i]$ but only the local parameters $x_i$ or the gradients $\nabla f_i(x_i)$.

To derive formal guarantees of convergence,
we require standard assumptions on convexity and smoothness of the local losses~\cite{Bastianello21ecc-proximalGradientMethod,Yuan16siamjo-convergenceDGD,Xin20spm-DecentralizedStochasticOptimization}.
In the following,
$\|\cdot\|$ denotes the $2$-norm.
\begin{definition}
	Function $f$ is $\mu$-strongly convex if there exists $\mu>0$ such that,
	for all $x$ and $y$,
	it holds
	\begin{equation}
		f(y) \ge f(x) + \nabla f(x)^\top(y - x) + \frac{\mu}{2} \norm{y - x}.
	\end{equation}
\end{definition}

\begin{definition}
	Function $f$ is $L$-smooth if $\nabla f$ is globally Lipschitz with Lipschitz constant $L>0$.
\end{definition}

\begin{ass}
	The local loss $f_i$ in~\eqref{eq:loss} is $\mu$-strongly convex and $L$-smooth for all $i \in \mathcal{V}$.
\end{ass}

\subsubsection*{Connection with federated learning}
The federated learning (FL) setup is the special case where each agent is connected to only one external agent,
called the aggregator.
This receives local models from agents,
combines them into a \emph{global model},
and sends the latter back to the agents,
which refine it with local data,
and the process repeats itself.
While the communication topology of federated learning is a simple star network centered at the aggregator,
this corresponds to a complete graph for~\eqref{eq:distr-learning} whereby the aggregator implements all communication links.

\subsection{What is Personalization in Distributed Learning?}\label{sec:what-personalization}

The distributed learning problem previously introduced treats all agents as if they are homogeneous,
\ie sample data from the same underlying distribution.
Indeed,
problem~\eqref{eq:distr-learning} forces all agents to learn a unique model under the implicit assumption that this will generalize to the distribution $\calP_i$ of each agent $i$.
However,
in practical cases,
the agents both have heterogeneous (non-i.i.d.) data and may want to prioritize their own local distributions.
For instance,
if industrial partners collaboratively train models for data-driven predictive maintenance,
each partner may want a model tailored to its own machines.
This necessity calls for \emph{personalization} of the parameter $x_i$ such that the model learned by agent $i$ does not just generalize to all $\calP_j$'s --- owing to the collaboration with other agents ---  but provides high accuracy specifically on the distribution $\calP_i$.
Of course,
there is a tradeoff between personalization and generalization that depends on inter-agent heterogeneity.

\subsubsection{Personalized federated learning}
In FL,
the global model broadcast by the aggregator can serve as a reference for personalization. 
If the global model differs significantly from a local model,
that agent may prioritize its local loss.
This can be accommodated as proposed in~\cite{Dinh20nips-personalizedFL}. Denoting the global parameters by $x$,
problem~\eqref{eq:distr-learning} is formally redefined as
\begin{subequations}\label{eq:personalization-fl}
	\begin{mini}
		{x\in\Real{n}}
		{\sum_{i\in\calV} F_i(x)}
		{\label{eq:prob-personalized}}
		{}
	\end{mini}
	where the local loss of agent $i$ is replaced by
    \begin{equation}\label{eq:personalized-problem}
		F_i(x) = \min_y \lb f_i(y) + \frac{\pi_i}{2} \norm[2]{y - x} \rb
    \end{equation}
    and the parameter $\pi_i$ quantifies the importance agent $i$ gives the global model.
\end{subequations}
The reformulation~\eqref{eq:personalization-fl} allows agent $i$ to personalize its parameter $x_i$ for its own dataset $\data[i]$ while controlling how different its local model and the global one are through the coefficient $\pi_i$.

\subsubsection{Personalized distributed learning}
The personalization-oriented reformulation~\eqref{eq:personalization-fl} relies on the hierarchical structure of FL,
whereby the aggregator computes and broadcasts the global model $x$ to all agents.
However,
an analogous formulation cannot be constructed for the fully distributed setup we consider,
where agents access only information sent by a few  neighbors.

\begin{prob}[Personalized distributed learning]\label{prob:personalization}
	Given the distributed learning problem~\eqref{eq:distr-learning},
	how to increase the local accuracy of each agent while maintaining high global accuracy?
	Formally,
	how to optimally trade between minimization of the local loss $f_i(x_i)$ of each agent $i$	and the total loss $\sum_{j\in\calV}f_j(x_j)$?
\end{prob}

\subsection{What is Resilience in Distributed Learning?}\label{sec:what-resilience}

The distributed learning setup~\eqref{eq:distr-learning} assumes that all agents are ``honest'' and correctly run the distributed algorithm.
However,
wireless communication allows external malicious agents to interfere and degrade the learned models.
This can be done by attacking (or acting as) an agent and spread misleading information during training,
\eg by sharing noisy parameters $x_i$ to lower accuracy.
Even without malicious attackers,
agents with many outliers or very different data distributions may degrade the distributed training of local models.

Standard approaches for privacy preservation include computationally expensive security mechanisms,
\eg encryption,
or ad-hoc solutions such as perturbing the transmitted parameters.
These approaches do little against disrupting model attacks.
Proactive cybersecurity mechanisms leverage the communication protocol to detect adversaries from transmitted packets,
but this can take quite some time,
neglects the content of messages,
and is not robust to normal agents with outliers in the training data.
Therefore,
there is a need to make distributed training not just secure but also \emph{resilient},
that is,
insensitive to disruptions introduced by adversaries or agents with poor training data without degrading the models learned by normal agents.

\begin{prob}[Resilient distributed learning]\label{prob:resilient-learning}
		Given the distributed learning problem~\eqref{eq:distr-learning},
		how to make the honest agents learn globally accurate parameters $x_i$ with unknown malicious agents?
		Formally,
		let $\calH\subset\calV$ denote honest agents,
		then we are interested in solving the problem
		\begin{mini}
			{\{x_i\}_{i\in\calH}}
			{\sum_{i\in\calH} f_i(x_i)}
			{\label{eq:distr-learning-resilient}}
			{}
			\addConstraint{x_i}{=x_j}{\ \forall i,j\in\calH.}
		\end{mini}
\end{prob}

\subsection{The Affinity Between Personalization and Resilience}\label{sec:affinity}

As motivated in the previous sections,
personalization and resilience respond to two different needs.
On the one hand,
an agent personalizes its own model to achieve high local accuracy,
possibly at the cost of lower accuracy on other agents' distributions.
On the other hand,
model attacks urge the honest agents to implement resilient mechanisms that prevent disruption of the cooperative training.

Nonetheless,
personalization and resilience share similarities.
They both reduce the emphasis on collaboration and focus on each individual agent to either enhance local accuracy (personalization) or reduce uncertainty in the source of received information (resilience).
In both scenarios each agent regards all other agents as potentially harmful to its own learning goal:
on the one hand,
the other agents may have different data distributions,
introducing undesired variance to the model;
on the other hand,
the identity of malicious agents is unknown,
forcing honest agents to be suspicious of their neighbors.

This discussion motivates the quest for a common mechanism that can accommodates both objectives,
which we propose and evaluate in the next sections.

\section{Algorithm Design}\label{sec:algorithm-design}

In this section,
we devise our distributed learning algorithm tailored to \cref{prob:personalization,prob:resilient-learning}.
We consider gradient-based algorithms consistently with a large portion of the machine learning literature.
Motivated by the discussion in \autoref{sec:affinity},
our design aims to optimally trade local accuracy (of each agent) for global accuracy (of all agents) to achieve personalized and resilient models.
We first compare gradient tracking (GT) and DGD in \autoref{sec:personalization-inaccuracy}.
DGD achieves higher local accuracy than GT and performs better with noisy training.
However,
this behavior of DGD naturally emerges as a by-product of its intrinsic suboptimality and cannot be tuned.
Therefore,
in \autoref{sec:personalization-fj} we combine DGD with the FJ model,
which captures stubbornness in opinion dynamics.
This model introduces an additional scalar parameter which we use to tune the learned models towards either local or global accuracy and account for variable requirements.

\subsection{DGD: Advantages of Inaccuracy}\label{sec:personalization-inaccuracy}

Consider the standard iterate of distributed gradient descent
\begin{equation}\label{eq:dgd}\tag{DGD}
	\b{x}_{k+1} = W\b{x}_k - \alpha\nabla f(\b{x}_k)
\end{equation}
where the bold symbol $\b{x}\in\Real{Nn}$ stacks the parameters of all agents and $f(\b{x}_k)$ stacks all local losses,
each evaluated at the corresponding agent's parameter.
The matrix $W$ is doubly stochastic and can be easily built in practice,
\eg with Metropolis weights.

It is well known that~\eqref{eq:dgd} does not solve problem~\eqref{eq:distr-learning} but achieves a neighborhood of the globally optimal model.
The size of this neighborhood depends on the heterogeneity of the local losses and is larger with more varied training data.
The next result quantifies the distance between the learned parameters and the optimum of~\eqref{eq:distr-learning}.

\begin{prop}\label{prop:dgd-fixed-point}
	Let $x^* = \argmin_{x \in \Real{n}} \sum_{i = 1}^N f_i(x)$ and $\bar{\b{x}}$ be the fixed point of \eqref{eq:dgd}.
	Then the following bound holds
	\begin{equation}\label{eq:dgd-gap}
		\norm{\bar{\b{x}} - \b{1} \otimes x^*} \leq \calO\!\lr c + \alpha D  \rr		
	\end{equation}
	where $c$ is a constant offset and
	\begin{equation}
		D = \sqrt{2 L \sum_{i = 1}^N \lr f_i(x^*) - f_i(x_i^*)\rr}
	\end{equation}
	with $x_i^* \doteq \argmin_{x \in \Real{n}} f_i(x)$.
\end{prop}
\begin{proof}
	By \cite[Theorem~7]{Yuan16siamjo-convergenceDGD} we know that \eqref{eq:dgd} converges to the point $\bar{\b{x}}$. Additionally, we have that
	$$
	\norm{\bar{\b{x}} - \b{1} \otimes x^*} \leq \calO\!\lr \alpha D' \rr
	$$
	where
	$$
	D' = \sqrt{2 L \sum_{i = 1}^N \lr f_i(0) - f_i(x_i^*)\rr}.
	$$
	Summing and subtracting $f_i(x^*)$, rearranging, and using $\sqrt{a + b} \leq \sqrt{a} + \sqrt{b}$ if $a, b \geq 0$, we get
	$$
	D' \leq c + D
	$$
	with $c \doteq \sqrt{2 L \sum_{i = 1}^N \lr f_i(0) - f_i(x^*)\rr}$.
\end{proof}

The constant $c$ in the bound~\eqref{eq:dgd-gap} does not depend on the cost functions.
On the other hand,
if the agents' losses are very heterogeneous,
the global minimizer $x^*$ and each local minimizer $x_i^*$ differ widely,
causing $D$ to increase.

This behavior of~\eqref{eq:dgd} is undesired in the nominal case where one aims to solve~\eqref{eq:distr-learning} and achieve global optimality.
It turns useful to improve on local accuracy of each agent instead,
because the updates~\eqref{eq:dgd} converge near both the global optimum $x^*$ and the minimizers of the local losses $f_i$.

\subsubsection*{Numerical evaluation}\label{sec:personalization-dgd-experiments}

\begin{table}
	\centering
	\caption{Accuracy with gradient tracking (ED) and distributed gradient descent (DGD and ATC).}
	\label{tab:personalization-dgd-v-ed-test}
	\begin{tabular}{ccccc}
		\toprule
		& Dataset & Min & Mean $\pm$ Std & Max \\
		\midrule
		\multirow{2}*{ED}  & global & $0.934$ & $0.934 \pm 0.000$     & $0.934$ \\
		& local 	& $0.640$ & $0.934 \pm 0.110$ & $1.000$ \\
		\midrule
		\multirow{2}*{\ref{eq:dgd}} & global & $0.848$ & $0.876 \pm 0.026$ & $0.918$ \\
		& local 	& $0.840$ & $0.974 \pm 0.054$ & $1.000$ \\
		\midrule
		\multirow{2}*{\ref{eq:atc}} & global & $0.868$ & $0.891 \pm 0.021$ & $0.928$ \\
		& local 	& $0.840$ & $0.966 \pm 0.063$ & $1.000$ \\
		\bottomrule
	\end{tabular}
\end{table}

\begin{table}
	\centering
	\caption{Accuracy with gradient tracking (ED) and distributed gradient descent (DGD and ATC) with noisy updates.}
	\label{tab:personalization-dgd-v-ed-test-noise}
	\begin{tabular}{ccccc}
		\toprule
		& Dataset & Min & Mean $\pm$ Std & Max \\
		\midrule
		\multirow{2}*{ED}  & global & $0.502$ & $0.506 \pm 0.003$     & $0.512$ \\
		& local 	& $0.020$ & $0.502 \pm 0.274$ & $1.000$ \\
		\midrule
		\multirow{2}*{\ref{eq:dgd}} & global & $0.604$ & $0.674 \pm 0.052$ & $0.758$ \\
		& local 	& $0.400$ & $0.686 \pm 0.194$ & $1.000$ \\
		\midrule
		\multirow{2}*{\ref{eq:atc}} & global & $0.520$ & $0.627 \pm 0.051$ & $0.704$ \\
		& local 	& $0.000$ & $0.676 \pm 0.320$ & $1.000$ \\
		\bottomrule
	\end{tabular}
\end{table}

To concretely assess potential advantages of DGD methods,
we consider a binary classification task on synthetic data.\footnote{
	Datasets, losses, and models are described in \autoref{sec:experiments}, \cref{ex:bin-class}.}
We compare~\eqref{eq:dgd} with exact diffusion (ED),
a gradient tracking algorithm proposed in~\cite{Yuan19tsp-ExactDiffusion} that solves~\eqref{eq:distr-learning},
and Adapt-Then-Combine (ATC), whose $k$th iterate is
\begin{equation}\label{eq:atc}\tag{ATC}
	\b{x}_{k+1} = W(\b{x}_k - \alpha\nabla f(\b{x}_k)).
\end{equation}
\blue{The classification accuracy of model $x$ on dataset $\data$ is defined as
\begin{equation}\label{eq:accuracy}
	\text{accuracy} \doteq 1 - \frac{|\{(\feat{}{},\lab{}{})\in\data:\hat{\lab{}{}}_{x}(\feat{}{})\neq\lab{}{}\}|}{|\data|},
\end{equation}
where $\hat{\lab{}{}}_{x}(\feat{}{})$ is the label output by model $x$ given features $\feat{}{}$.
\Cref{tab:personalization-dgd-v-ed-test} reports accuracy statistics across the $N$ agents.
For agent $i$,
local accuracy is computed on test dataset $\data[i]^{\textnormal{test}}$ and global accuracy on combined dataset $\cup_{i\in\calV}\data[i]^\textnormal{test}$.
The values show that~\eqref{eq:dgd} outperforms ED on local accuracy,
because the latter converges to the minimizer of the global loss and cannot accommodate for heterogeneous datasets.
Algorithm~\eqref{eq:atc} places between~\eqref{eq:dgd} and ED in terms of average accuracy.}

Notably,
gradient tracking is sensitive to noise.
In distributed learning,
this clashes against noisy gradients (\eg in stochastic gradient descent),
wireless channel erasure,
or messages sent by malicious agents.
\blue{\cref{tab:personalization-dgd-v-ed-test-noise} reports accuracy scores computed as in \cref{tab:personalization-dgd-v-ed-test} but with algorithm updates additively perturbed by i.i.d. Gaussian noise drawn from $\calN(0,2)$.}
The accuracy of ED drops significantly,
whereas~\eqref{eq:dgd} and~\eqref{eq:atc} are more robust as shown by the smaller decreases in average accuracy.

\blue{We conclude that distributed gradient descent methods improve over exact gradient tracking both on local accuracy and with noisy updates,
suggesting opportunities for enhancing personalization and resilience in light of the discussion in \autoref{sec:affinity}.}
However,
neither~\eqref{eq:atc} nor~\eqref{eq:dgd} provide ways to compare local models with a ``global'' consensus-based model,
as opposed to gradient tracking,
and offer no easy way to tune local or global accuracy reached by each agent. 
We next address this issue through an opinion dynamics model that structurally accounts for heterogeneous agents.

\subsection{FJ-DGD: Harnessing Stubbornness for Local Accuracy}\label{sec:personalization-fj}

We draw inspiration from the Friedkin-Johnsen (FJ) model to tune between global and local accuracy achieved by learned models and enhance personalization and resilience.
The original FJ model tracks dynamic evolution of opinions $\b{x}$ starting from an initial condition $\b{x}_0$,
assuming that the agents are partially stubborn and retain their initial opinions throughout~\cite{Friedkin90jms-fjmodel,Proskurnikov17arc-tutorialModelingSocialNet}:
\begin{equation}\label{eq:fj}
	\b{x}_{k+1} = (I - \Lambda)W\b{x}_k + \Lambda\b{x}_0.
\end{equation}
Matrix $\Lambda$ is diagonal and its $i$th diagonal element $\lambda_i\in[0,1]$ represents the stubbornness of agent $i$.
The case with $\lambda_i \equiv 0$ reduces~\eqref{eq:fj} to the consensus algorithm,
whereas arbitrary stubbornness parameters $\lambda_i>0$ prevent a consensus apart from trivial cases.
In words,
agents behaving according to~\eqref{eq:fj} embed the others’ opinions (through the consensus term $W\b{x}_k$) but always mediate with their own initial opinion.
The closer $\lambda_i$ is to $1$,
the less the final opinion $x_i$ is affected by the others.

\blue{This offers a simple workaround \blue{to improve local accuracy} by letting each agent retain local information at all times.
	Intuitively,
	the FJ model can purposely bias the local model $x_i$ of agent $i$ towards the optimal one for local data $\data[i]$,
	whereby the closer $\lambda_i$ to $1$,
	the stronger such as bias.}
A straightforward implementation of this intuition yields the following algorithm
\begin{equation}\label{eq:fj-dgd-naive}
	\b{x}_{k+1} = \Lambda\b{x}^* + (I - \Lambda)(W\b{x}_k + \alpha\nabla f(\b{x}_k)),
\end{equation}
where $\b{x}^* \doteq [(x_1^*)^\top, \ldots, (x_N^*)^\top]^\top$ and $x_i^* = \argmin_{x} f_i(x)$ is the local optimizer of agent $i$.
This strategy requires the agents to spend extra time to pre-compute the local optimal parameters $x_i^*$ before collaborative training,
with the added complication of synchronizing the start.
To avoid this issue,
we let the agents concurrently update the local models during collaborative training.
This results in the following algorithm,
whereby agent $i$ tracks it's local minimizer $x_i^*$ with variable $y_i$
\begin{equation}\label{eq:fj-dgd}\tag{FJ-DGD-1}
	\begin{aligned}
		\b{y}_{k+1} &= \b{y}_{k} - \alpha\nabla f(\b{y}_k), \qquad \b{y}_{0} = \b{x}_{0}\\
		\b{x}_{k+1} &= \Lambda\b{y}_{k+1} + (I - \Lambda)(W\b{x}_k-\alpha\nabla f(\b{x}_k)).
	\end{aligned}
\end{equation}
Since $\b{y}_k$ converges to $\b{x}^*$,
both~\eqref{eq:fj-dgd-naive} and~\eqref{eq:fj-dgd} have the same fixed point.
Algorithm~\eqref{eq:fj-dgd} combines the consensus-based aggregation of DGD and the stubbornness of the FJ model to learn models $ x_i$, 
which are influenced by both  behaviors throughout the training. 
Alternatively, 
the stubborn and consensus-based models may be separately computed and combined at a second stage.
This yields the following algorithm,
whereby each agent $i$ uses two extra variables $y_i$ and $z_i$
\begin{equation}\label{eq:fj-dgd-2}\tag{FJ-DGD-2}
	\begin{aligned}
		\b{y}_{k+1} &= \b{y}_{k} - \alpha\nabla f(\b{y}_k), \qquad \b{y}_{0} = \b{x}_{0}\\
		\b{z}_{k+1} &= W\b{z}_k-\alpha\nabla f(\b{z}_k), \qquad \b{z}_{0} = \b{x}_{0}\\
		\b{x}_{k+1} &= \Lambda\b{y}_{k+1} + (I - \Lambda)\b{z}_{k+1}.
	\end{aligned}
\end{equation}
In \cref{sec:convergence-analysis} we  analyze  convergence of~\eqref{eq:fj-dgd-2}.
This choice is due to the higher flexibility of~\eqref{eq:fj-dgd-2}, whose $\b{y}$- and $\b{z}$-updates are independent.
A similar analysis can be carried out for~\eqref{eq:fj-dgd}.

\subsection{FJ-DGD With Corrupted Updates}
\label{sec:resilience}

The previous section considers a nominal collaborative scenario where all agents truthfully obey the designed algorithm.
We now extend this case to the scenario where some agents do not exactly follow the designed update rule.
In particular,
this model can capture malicious agents that transmit noisy or deceiving models to disrupt the collaborative training carried out by the other agents.
In light of the affinity between personalization and resilience discussed in \autoref{sec:affinity},
FJ-DGD can be tuned to achieve good global accuracy in this case.
Intuitively,
if the agents reduce collaboration in a controlled manner,
they mitigate the effect of malicious or low-quality information they receive.
In the following,
we focus on~\eqref{eq:fj-dgd-2} for the sake of conciseness.

Let $\mathcal{H}$ and $\mathcal{M}$ be the sub-sets of honest and malicious agents, respectively,
such that $\mathcal{H} \cap \mathcal{M} = \emptyset$ and $\mathcal{V} = \mathcal{H} \cup \mathcal{M}$.
We assume that malicious agents participating in the learning process freely choose $\b{z}$ in the updates of~\eqref{eq:fj-dgd-2} to \blue{degrade performance of honest agents' local models}.
Define the vector $\b{e}_k$ such that
\begin{equation}
	[ \b{e}_k ]_i = 
	\begin{cases}
		0 & \text{if} \ i \in \mathcal{H} \\
		e_{i,k} & \text{if} \ i \in \mathcal{M}.
	\end{cases}
\end{equation}
\blue{Error vector $\b{e}_k$ represents the deviation from nominal updates caused by malicious behavior.}
From the perspective of agent $i\in\calH$ \blue{(which ignores whether $e_{j,k}=0$ or $e_{j,k}\neq0$ for any neighbor $j$)},
the updates read
\begin{multline}
	z_{i,k+1} = w_{ii} z_{i,k} 
	+ \sum_{j \in \mathcal{N}_i \cap \mathcal{H}} w_{ij} z_{j,k} \\
	+ \sum_{j \in \mathcal{N}_i \cap \mathcal{M}} w_{ij} (z_{j,k} + e_{j,k}) 
	- \alpha \nabla f_i(z_{i,k}),
\end{multline}
and replacing the $\b{z}$-update of~\eqref{eq:fj-dgd-2} yields
\begin{equation}\label{eq:fj-dgd-noise}\tag{FJ-DGD-N}
	\begin{aligned}
		\b{y}_{k+1} &= \b{y}_{k} - \alpha\nabla f(\b{y}_k), \quad \b{y}_{0} = \b{x}_{0}\\
		\b{z}_{k+1} &= W\b{z}_k-\alpha\nabla f(\b{z}_k) + W \b{e}_k, \quad \b{z}_{0} = \b{x}_{0}\\
		\b{x}_{k+1} &= \Lambda\b{y}_{k+1} + (I - \Lambda)\b{z}_{k+1}.
	\end{aligned}
\end{equation}
This modified version behaves worse because of the unmodeled disturbances in $\b{e}_k$,
which may not relate to the algorithm itself.
The following section is dedicated to the performance analysis of~\eqref{eq:fj-dgd-2} and of its corrupted version~\eqref{eq:fj-dgd-noise}.

\section{Convergence Analysis}
\label{sec:convergence-analysis}

We propose two main results.
The first one quantifies the final models learned by~\eqref{eq:fj-dgd-2} within a fully collaborative setup,
along with the speed of convergence.
This gives an indication of how long training is required.
Recall that we use the notation $\b{x}^* \doteq [(x_1^*)^\top, \ldots, (x_N^*)^\top]^\top$, with $x_i^* = \argmin_{x} f_i(x)$ the local optimizer of agent $i$.

\begin{thm}\label{prop:convergence-fj-dgd}
	Let $\bar{\b{x}}$ be the fixed point of~\eqref{eq:dgd}.
	Algorithm~\eqref{eq:fj-dgd-2} converges linearly to $\Lambda \b{x}^* + (I - \Lambda) \bar{\b{x}}$ at a rate of
	\begin{equation}\label{eq:conv-rate}
		\zeta = \max\left\{ |1 - \alpha \mu|, |1 - \alpha L|, |\lambda_{\min}(W) - \alpha L| \right\}
	\end{equation}
	where $\lambda_{\min}(W)$ is the smallest eigenvalue of $W$.
    \begin{proof}
        The result is a consequence of the following facts:
        \begin{itemize}
            \item $\b{y}_{k+1} = \b{y}_{k} - \alpha\nabla f(\b{y}_k)$ converges linearly to $\b{x}^*$, at a rate $\zeta' = \max\left\{ |1 - \alpha \mu|, |1 - \alpha L| \right\}$;
            
            \item \eqref{eq:dgd} converges linearly to $\bar{\b{x}}$, at a rate $\zeta'' = \max\left\{ |1 - \alpha \mu|, |\lambda_{\min}(W) - \alpha L|\right\}$.
        \end{itemize}
        Combining the items above yields the rate in~\eqref{eq:conv-rate}.
    \end{proof}
\end{thm}

\cref{prop:convergence-fj-dgd} shows that algorithm~\eqref{eq:fj-dgd-2} enjoys a geometric convergence rate.
The final models are convex combinations between the global optimum $x^*$ and the models $\bar{x}_i$ learned with~\eqref{eq:dgd}.
This allows a designer to easily tune the behavior of~\eqref{eq:fj-dgd-2} based on how much local accuracy is preferred over global accuracy.
In \autoref{sec:experiments},
we show how to use this to enforce personalization or help resilience. 
We also remark that the choice of $\Lambda$ here does not affect the convergence rate of the algorithm,
but only its fixed point.
Indeed, the convergence rate is fully characterized by the separate convergence rates of the $\b{y}$- and $\b{z}$-updates, which do not depend on $\Lambda$.
The fixed point though does depend on it, as it is characterized as the convex combination $\Lambda \b{x}^* + (I - \Lambda) \bar{\b{x}}$.

The next result evaluates the convergence of~\eqref{eq:fj-dgd-noise} with noisy terms $e_{i,k}$ injected by malicious agents.
\blue{We assume that these cannot transmit arbitrary messages and bound their magnitude.
This assumption is not restrictive because,
if honest agents receive inappropriately large model weights,
they can easily detect malicious agents and exclude them from training.}

\begin{thm}\label{thm:convergence-noisy}
    Let $\| \b{e}_k \| \leq \tau$.
	Algorithm~\eqref{eq:fj-dgd-noise} converges to a neighborhood of $\hat{\b{x}} = \Lambda \b{x}^* + (I - \Lambda) \bar{\b{x}}$ characterized by
	\begin{equation}\label{eq:neighborhood}
		\| \b{x}_k - \hat{\b{x}} \| \leq \zeta^k \| \b{x}_0 - \hat{\b{x}} \| + (1 - \min_i \lambda_i) \tau \sum_{h = 0}^{k-1} \zeta^{k - h - 1}
	\end{equation}
    at a rate of $\zeta$, defined in~\eqref{eq:conv-rate}.
    \begin{proof}
        The goal is to provide a bound to the distance $\| \b{x}_k - \b{\hat{x}} \|$, with $\Lambda \b{x}^* + (I - \Lambda) \bar{\b{x}}$ the fixed point of~\eqref{eq:fj-dgd-2};
        see \cref{prop:convergence-fj-dgd}.
        Using the characterization of the update~\eqref{eq:fj-dgd-noise} and the triangle inequality yields
        \begin{align*}
            \| \b{x}_{k+1} - \b{\hat{x}} \| &\leq \| \Lambda \b{y}_{k+1} + (I-\Lambda) \left(W \b{z}_k - \alpha \nabla f(\b{z}_k) \right) - \b{\hat{x}} \| \\ & \quad + \| (I-\Lambda) W \b{e}_k \| \\
            &\overset{(i)}{\leq} \zeta \| \b{x}_k - \b{\hat{x}} \| + \| (I-\Lambda) W \b{e}_k \|
        \end{align*}
        where $(i)$ holds by contractiveness of~\eqref{eq:fj-dgd-2}.
        Finally, by sub-multiplicativity of the norm we have $\| (I-\Lambda) W \b{e}_k \| \leq \| (I-\Lambda)\| \| W \| \| \b{e}_k \|$ and using $\| I-\Lambda \| = 1 - \min_i \lambda_i$, $\| W \| = 1$, $\| \b{e}_k \| \leq \tau$, we have
        \begin{equation}\label{eq:error-noisy-update}
            \| \b{x}_{k+1} - \b{\hat{x}} \| \leq \zeta \| \b{x}_k - \b{\hat{x}} \| + (1 - \min_i \lambda_i) \tau,
        \end{equation}
        and iterating~\eqref{eq:error-noisy-update} over time yields the result.
    \end{proof}
\end{thm}

\Cref{thm:convergence-noisy} proves that~\eqref{eq:fj-dgd-noise} converges inexactly, due to the noise injected by the malicious agents, to a neighborhood whose radius is upper bounded by $(1 - \min_i \lambda_i) \tau / (1 - \zeta)$.
Honest agents can regulate their $\lambda_i$'s to attenuate the impact of malicious agents, with the extreme choice of $\lambda_i = 1$ leading to all agents isolating and performing local training only.

\section{Numerical Experiments}
\label{sec:experiments}

We now test our algorithm on distributed learning tasks where we wish to impose personalization and resilience.
While these two needs share similarities as previously discussed,
they correspond to two different scenarios --- fully collaborative heterogeneous agents \emph{vs.} some non-collaborative or adversarial agents whose identity is unknown.
Therefore,
we perform two different numerical tests.
Also, this allows us to better identify the effects of our algorithm for personalization and resilience.

\subsection{Evaluating Personalization}\label{sec:personalization-fj-experiments}

\begin{figure*}
	\centering
	\subfloat[Local training accuracy.
	\label{fig:fj-djd-loc-train}]{\resizebox{.25\textwidth}{!}{

\begin{tikzpicture}[yscale=1.07]
	
	\definecolor{color0}{rgb}{0.298039215686275,0.447058823529412,0.690196078431373}
	\definecolor{color1}{rgb}{0.866666666666667,0.517647058823529,0.32156862745098}
	\definecolor{color2}{rgb}{0.333333333333333,0.658823529411765,0.407843137254902}
	\definecolor{color3}{rgb}{0.768627450980392,0.305882352941176,0.32156862745098}
	\definecolor{color4}{rgb}{0.505882352941176,0.447058823529412,0.701960784313725}
	
	\definecolor{darkgray176}{RGB}{176,176,176}
	\definecolor{green01270}{RGB}{0,127,0}
	\definecolor{lightgray204}{RGB}{204,204,204}
	
	\begin{axis}[
		height=5cm,
		legend cell align={left},
		legend style={
			fill opacity=0.9,
			draw opacity=1,
			text opacity=1,
			at={(1.01,0)},
			anchor=south east,
			draw=white!80!black,
			font=\scriptsize
		},
		legend columns=3,
		tick pos=left,
		unbounded coords=jump,
		xtick={-.2,.8,2,3.3},
		xtick style={color=white},
		xticklabels={DGD, Local GD, FJ-DGD-1, FJ-DGD-2},
		xticklabel style={font=\scriptsize},
		y grid style={white!69.0196078431373!black},
		ylabel={{Accuracy}},
		ylabel shift=-5pt,
		ylabel style={font=\footnotesize},
		ymajorgrids,
		ytick={.8,.9,1},
		ymax=1.002,
		ytick style={color=black},
		yticklabel style={font=\scriptsize},
		ytick align=inside,
		]
		\addplot [line width=1.08pt, color0, mark=*, mark size=3, mark options={solid}, only marks, forget plot]
		table {%
			0 0.9217777777777778
		};
		\addplot [line width=1.08pt, color0, forget plot]
		table {%
			-0.1 0.813
			0.1 0.813
			nan nan
			0 0.813
			0 1
			nan nan
			-0.1 1
			0.1 1
		};
		\addplot [line width=1.08pt, color0, mark=*, mark size=3, mark options={solid}, only marks, forget plot]
		table {%
			1 0.9433333333333334
		};
		\addplot [line width=1.08pt, color0, forget plot]
		table {%
			0.9 0.871
			1.1 0.871
			nan nan
			1 0.871
			1 1
			nan nan
			0.9 1
			1.1 1
		};
		\addplot [line width=1.08pt, color1, mark=square*, mark size=3, mark options={solid}, only marks]
		table {%
			2 0.9351111111111111
		};
		\addlegendentry{$\lambda=0.25$}
		\addplot [line width=1.08pt, color1, forget plot]
		table {%
			1.9 0.849
			2.1 0.849
			nan nan
			2 0.849
			2 1
			nan nan
			1.9 1
			2.1 1
		};
		\addplot [line width=1.08pt, color2, mark=diamond*, mark size=4, mark options={solid}, only marks]
		table {%
			2.2 0.9400000000000001
		};
		\addlegendentry{$\lambda=0.5$}
		\addplot [line width=1.08pt, color2, forget plot]
		table {%
			2.1 0.861
			2.3 0.861
			nan nan
			2.2 0.861
			2.2 1
			nan nan
			2.1 1
			2.3 1
		};
		\addplot [line width=1.08pt, color3, mark=pentagon*, mark size=3.5, mark options={solid}, only marks]
		table {%
			2.4 0.9417777777777777
		};
		\addlegendentry{$\lambda=0.75$}
		\addplot [line width=1.08pt, color3, forget plot]
		table {%
			2.3 0.867
			2.5 0.867
			nan nan
			2.4 0.867
			2.4 1
			nan nan
			2.3 1
			2.5 1
		};
		\addplot [line width=1.08pt, color1, mark=square*, mark size=3, mark options={solid}, only marks]
		table {%
			3 0.9286666666666668
		};
		\addplot [line width=1.08pt, color1, forget plot]
		table {%
			2.9 0.832
			3.1 0.832
			nan nan
			3 0.832
			3 1
			nan nan
			2.9 1
			3.1 1
		};
		\addplot [line width=1.08pt, color2, mark=diamond*, mark size=4, mark options={solid}, only marks]
		table {%
			3.2 0.9351111111111111
		};
		\addplot [line width=1.08pt, color2, forget plot]
		table {%
			3.1 0.85
			3.3 0.85
			nan nan
			3.2 0.85
			3.2 1
			nan nan
			3.1 1
			3.3 1
		};
		\addplot [line width=1.08pt, color3, mark=pentagon*, mark size=3.5, mark options={solid}, only marks]
		table {%
			3.4  0.9422222222222223
		};
		\addplot [line width=1.08pt, color3, forget plot]
		table {%
			3.3 0.869
			3.5 0.869
			nan nan
			3.4 0.869
			3.4 1
			nan nan
			3.3 1
			3.5 1
		};
	\end{axis}
	
\end{tikzpicture}}}
	\subfloat[Global training accuracy.
	\label{fig:fj-djd-glob-train}]{\resizebox{.25\textwidth}{!}{

\begin{tikzpicture}
	
	\definecolor{color0}{rgb}{0.298039215686275,0.447058823529412,0.690196078431373}
	\definecolor{color1}{rgb}{0.866666666666667,0.517647058823529,0.32156862745098}
	\definecolor{color2}{rgb}{0.333333333333333,0.658823529411765,0.407843137254902}
	\definecolor{color3}{rgb}{0.768627450980392,0.305882352941176,0.32156862745098}
	\definecolor{color4}{rgb}{0.505882352941176,0.447058823529412,0.701960784313725}
	
	\begin{axis}[
		height=5cm,
		tick pos=left,
		unbounded coords=jump,
		xtick={-.2,.8,2,3.3},
		xtick style={color=white},
		xticklabels={DGD, Local GD, FJ-DGD-1, FJ-DGD-2},
		xticklabel style={font=\scriptsize},
		y grid style={white!69.0196078431373!black},
		ymajorgrids,
		ytick={.4, .5, .6, .7,.8,.9},
		ytick style={color=black},
		yticklabel style={font=\scriptsize},
		ytick align=inside,
		]
		\addplot [line width=1.08pt, color0, mark=*, mark size=3, mark options={solid}, only marks, forget plot]
		table {%
			0 0.915
		};
		\addplot [line width=1.08pt, color0, forget plot]
		table {%
			-0.1 0.912
			0.1 0.912
			nan nan
			0 0.912
			0 0.9176
			nan nan
			-0.1 0.9176
			0.1 0.9176
		};
		\addplot [line width=1.08pt, color0, mark=*, mark size=3, mark options={solid}, only marks, forget plot]
		table {%
			1 0.592
		};
		\addplot [line width=1.08pt, color0, forget plot]
		table {%
			0.9 0.482
			1.1 0.482
			nan nan
			1 0.482
			1 0.701
			nan nan
			0.9 0.701
			1.1 0.701
		};
		\addplot [line width=1.08pt, color1, mark=square*, mark size=3, mark options={solid}, only marks]
		table {%
			2 0.678
		};
		\addplot [line width=1.08pt, color1, forget plot]
		table {%
			1.9 0.591
			2.1 0.591
			nan nan
			2 0.591
			2 0.765
			nan nan
			1.9 0.765
			2.1 0.765
		};
		\addplot [line width=1.08pt, color2, mark=diamond*, mark size=4, mark options={solid}, only marks]
		table {%
			2.2 0.632
		};
		\addplot [line width=1.08pt, color2, forget plot]
		table {%
			2.1 0.527
			2.3 0.527
			nan nan
			2.2 0.527
			2.2 0.737
			nan nan
			2.1 0.737
			2.3 0.737
		};
		\addplot [line width=1.08pt, color3, mark=pentagon*, mark size=3.5, mark options={solid}, only marks]
		table {%
			2.4 0.606
		};
		\addplot [line width=1.08pt, color3, forget plot]
		table {%
			2.3 0.5
			2.5 0.5
			nan nan
			2.4 0.5
			2.4 0.713
			nan nan
			2.3 0.713
			2.5 0.713
		};
		\addplot [line width=1.08pt, color1, mark=square*, mark size=3, mark options={solid}, only marks]
		table {%
			3 0.823
		};
		\addplot [line width=1.08pt, color1, forget plot]
		table {%
			2.9 0.770
			3.1 0.770
			nan nan
			3 0.770
			3 0.877
			nan nan
			2.9 0.877
			3.1 0.877
		};
		\addplot [line width=1.08pt, color2, mark=diamond*, mark size=4, mark options={solid}, only marks]
		table {%
			3.2 0.723
		};
		\addplot [line width=1.08pt, color2, forget plot]
		table {%
			3.1 0.643
			3.3 0.643
			nan nan
			3.2 0.643
			3.2 0.804
			nan nan
			3.1 0.804
			3.3 0.804
		};
		\addplot [line width=1.08pt, color3, mark=pentagon*, mark size=3.5, mark options={solid}, only marks]
		table {%
			3.4  0.646
		};
		\addplot [line width=1.08pt, color3, forget plot]
		table {%
			3.3 0.546
			3.5 0.546
			nan nan
			3.4 0.546
			3.4 0.745
			nan nan
			3.3 0.745
			3.5 0.745
		};
	\end{axis}
	
\end{tikzpicture}}}
	\subfloat[Local test accuracy.
	\label{fig:fj-djd-loc-test}]{\resizebox{.25\textwidth}{!}{

\begin{tikzpicture}
	
	\definecolor{color0}{rgb}{0.298039215686275,0.447058823529412,0.690196078431373}
	\definecolor{color1}{rgb}{0.866666666666667,0.517647058823529,0.32156862745098}
	\definecolor{color2}{rgb}{0.333333333333333,0.658823529411765,0.407843137254902}
	\definecolor{color3}{rgb}{0.768627450980392,0.305882352941176,0.32156862745098}
	\definecolor{color4}{rgb}{0.505882352941176,0.447058823529412,0.701960784313725}
	
	\begin{axis}[
		height=5cm,
		tick pos=left,
		unbounded coords=jump,
		xtick={-.2,.8,2,3.3},
		xtick style={color=white},
		xticklabels={DGD, Local GD, FJ-DGD-1, FJ-DGD-2},
		xticklabel style={font=\scriptsize},
		y grid style={white!69.0196078431373!black},
		ymajorgrids,
		ymax=1.002,
		ytick style={color=black},
		yticklabel style={font=\scriptsize},
		ytick align=inside,
		]
		\addplot [line width=1.08pt, color0, mark=*, mark size=3, mark options={solid}, only marks, forget plot]
		table {%
			0 0.908
		};
		\addplot [line width=1.08pt, color0, forget plot]
		table {%
			-0.1 0.788
			0.1 0.788
			nan nan
			0 0.788
			0 1
			nan nan
			-0.1 1
			0.1 1
		};
		\addplot [line width=1.08pt, color0, mark=*, mark size=3, mark options={solid}, only marks, forget plot]
		table {%
			1 0.926
		};
		\addplot [line width=1.08pt, color0, forget plot]
		table {%
			0.9 0.836
			1.1 0.836
			nan nan
			1 0.836
			1 1
			nan nan
			0.9 1
			1.1 1
		};
		\addplot [line width=1.08pt, color1, mark=square*, mark size=3, mark options={solid}, only marks]
		table {%
			2 0.920
		};
		\addplot [line width=1.08pt, color1, forget plot]
		table {%
			1.9 0.824
			2.1 0.824
			nan nan
			2 0.824
			2 1
			nan nan
			1.9 1
			2.1 1
		};
		\addplot [line width=1.08pt, color2, mark=diamond*, mark size=4, mark options={solid}, only marks]
		table {%
			2.2 0.924
		};
		\addplot [line width=1.08pt, color2, forget plot]
		table {%
			2.1 0.834
			2.3 0.834
			nan nan
			2.2 0.834
			2.2 1
			nan nan
			2.1 1
			2.3 1
		};
		\addplot [line width=1.08pt, color3, mark=pentagon*, mark size=3.5, mark options={solid}, only marks]
		table {%
			2.4 0.928
		};
		\addplot [line width=1.08pt, color3, forget plot]
		table {%
			2.3 0.843
			2.5 0.843
			nan nan
			2.4 0.843
			2.4 1
			nan nan
			2.3 1
			2.5 1
		};
		\addplot [line width=1.08pt, color1, mark=square*, mark size=3, mark options={solid}, only marks]
		table {%
			3 0.920
		};
		\addplot [line width=1.08pt, color1, forget plot]
		table {%
			2.9 0.822
			3.1 0.822
			nan nan
			3 0.822
			3 1
			nan nan
			2.9 1
			3.1 1
		};
		\addplot [line width=1.08pt, color2, mark=diamond*, mark size=4, mark options={solid}, only marks]
		table {%
			3.2 0.926
		};
		\addplot [line width=1.08pt, color2, forget plot]
		table {%
			3.1 0.838
			3.3 0.838
			nan nan
			3.2 0.838
			3.2 1
			nan nan
			3.1 1
			3.3 1
		};
		\addplot [line width=1.08pt, color3, mark=pentagon*, mark size=3.5, mark options={solid}, only marks]
		table {%
			3.4  0.93
		};
		\addplot [line width=1.08pt, color3, forget plot]
		table {%
			3.3 0.848
			3.5 0.848
			nan nan
			3.4 0.848
			3.4 1
			nan nan
			3.3 1
			3.5 1
		};
	\end{axis}
	
\end{tikzpicture}}}
	\subfloat[Global test accuracy.
	\label{fig:fj-djd-glob-test}]{\resizebox{.25\textwidth}{!}{

\begin{tikzpicture}
	
	\definecolor{color0}{rgb}{0.298039215686275,0.447058823529412,0.690196078431373}
	\definecolor{color1}{rgb}{0.866666666666667,0.517647058823529,0.32156862745098}
	\definecolor{color2}{rgb}{0.333333333333333,0.658823529411765,0.407843137254902}
	\definecolor{color3}{rgb}{0.768627450980392,0.305882352941176,0.32156862745098}
	\definecolor{color4}{rgb}{0.505882352941176,0.447058823529412,0.701960784313725}
	
	\begin{axis}[
		height=5cm,
		tick pos=left,
		unbounded coords=jump,
		xtick={-.2,.8,2,3.3},
		xtick style={color=white},
		xticklabels={DGD, Local GD, FJ-DGD-1, FJ-DGD-2},
		xticklabel style={font=\scriptsize},
		y grid style={white!69.0196078431373!black},
		ymajorgrids,
		ytick={.4, .5, .6, .7,.8,.9},
		ytick style={color=black},
		yticklabel style={font=\scriptsize},
		ytick align=inside,
		]
		\addplot [line width=1.08pt, color0, mark=*, mark size=3, mark options={solid}, only marks, forget plot]
		table {%
			0 0.907
		};
		\addplot [line width=1.08pt, color0, forget plot]
		table {%
			-0.1 0.904
			0.1 0.904
			nan nan
			0 0.904
			0 0.910
			nan nan
			-0.1 0.910
			0.1 0.910
		};
		\addplot [line width=1.08pt, color0, mark=*, mark size=3, mark options={solid}, only marks, forget plot]
		table {%
			1 0.588
		};
		\addplot [line width=1.08pt, color0, forget plot]
		table {%
			0.9 0.484
			1.1 0.484
			nan nan
			1 0.484
			1 0.691
			nan nan
			0.9 0.691
			1.1 0.691
		};
		\addplot [line width=1.08pt, color1, mark=square*, mark size=3, mark options={solid}, only marks]
		table {%
			2 0.684
		};
		\addplot [line width=1.08pt, color1, forget plot]
		table {%
			1.9 0.602
			2.1 0.602
			nan nan
			2 0.602
			2 0.766
			nan nan
			1.9 0.766
			2.1 0.766
		};
		\addplot [line width=1.08pt, color2, mark=diamond*, mark size=4, mark options={solid}, only marks]
		table {%
			2.2 0.630
		};
		\addplot [line width=1.08pt, color2, forget plot]
		table {%
			2.1 0.526
			2.3 0.526
			nan nan
			2.2 0.526
			2.2 0.733
			nan nan
			2.1 0.733
			2.3 0.733
		};
		\addplot [line width=1.08pt, color3, mark=pentagon*, mark size=3.5, mark options={solid}, only marks]
		table {%
			2.4 0.600
		};
		\addplot [line width=1.08pt, color3, forget plot]
		table {%
			2.3 0.496
			2.5 0.496
			nan nan
			2.4 0.496
			2.4 0.704
			nan nan
			2.3 0.704
			2.5 0.704
		};
		\addplot [line width=1.08pt, color1, mark=square*, mark size=3, mark options={solid}, only marks]
		table {%
			3 0.816
		};
		\addplot [line width=1.08pt, color1, forget plot]
		table {%
			2.9 0.758
			3.1 0.758
			nan nan
			3 0.758
			3 0.874
			nan nan
			2.9 0.874
			3.1 0.874
		};
		\addplot [line width=1.08pt, color2, mark=diamond*, mark size=4, mark options={solid}, only marks]
		table {%
			3.2 0.719
		};
		\addplot [line width=1.08pt, color2, forget plot]
		table {%
			3.1 0.636
			3.3 0.636
			nan nan
			3.2 0.636
			3.2 0.801
			nan nan
			3.1 0.801
			3.3 0.801
		};
		\addplot [line width=1.08pt, color3, mark=pentagon*, mark size=3.5, mark options={solid}, only marks]
		table {%
			3.4  0.641
		};
		\addplot [line width=1.08pt, color3, forget plot]
		table {%
			3.3 0.546
			3.5 0.546
			nan nan
			3.4 0.546
			3.4 0.737
			nan nan
			3.3 0.737
			3.5 0.737
		};
	\end{axis}
	
\end{tikzpicture}}}
	\caption{Accuracy with DGD and its FJ-based variants for the task in \cref{ex:bin-class}.
		Marks show the mean and bars one standard deviation intervals across all agents.}
	\label{fig:fj-dgd}
\end{figure*}
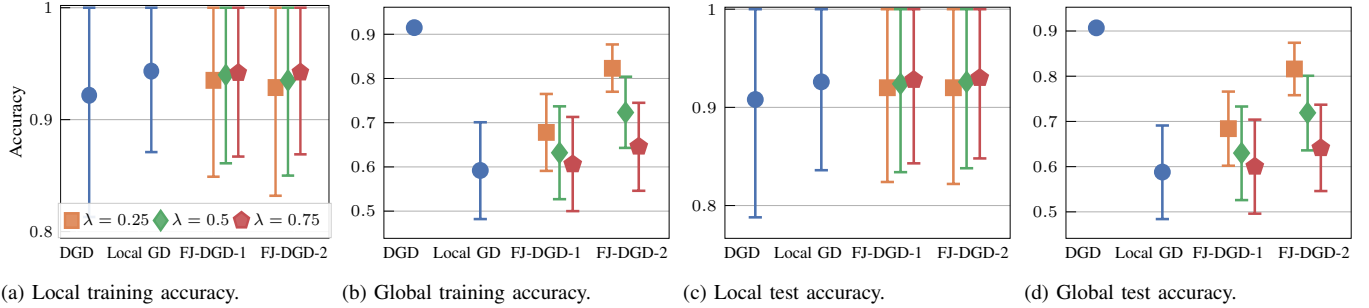

\begin{figure}
	\centering

\begin{tikzpicture}
	
	\definecolor{color0}{rgb}{0.298039215686275,0.447058823529412,0.690196078431373}
	\definecolor{color1}{rgb}{0.866666666666667,0.517647058823529,0.32156862745098}
	\definecolor{color2}{rgb}{0.333333333333333,0.658823529411765,0.407843137254902}
	\definecolor{color3}{rgb}{0.768627450980392,0.305882352941176,0.32156862745098}
	\definecolor{color4}{rgb}{0.505882352941176,0.447058823529412,0.701960784313725}
	\definecolor{lightgray204}{RGB}{204,204,204}
	
	\begin{axis}[
		height=4.5cm,
		width=\columnwidth,
		legend cell align={left},
		legend style={fill opacity=0.9,
			draw opacity=1,
			text opacity=1,
			draw=lightgray204,
			draw=white!80!black,
			font=\small},
		legend pos=south east,
		tick pos=left,
		x grid style={white!69.0196078431373!black},
		xlabel={$\lambda$},
		xlabel style={font=\small},
		xmajorgrids,
		xtick={0.25,0.5,0.75},
		xtick style={color=black},
		y grid style={white!69.0196078431373!black},
		ylabel={Accuracy},
		ylabel style={font=\small},
		ymin=.649, ymax=.801,
		ymajorgrids,
		ytick style={color=black},
		yticklabel style={font=\small},
		]
		\addplot+[very thick, black, no marks]
		table {%
			0.25 0.68
			0.5 0.68
			0.75 0.68
		};
		\addlegendentry{DGD}
		\addplot+[very thick, black, dashed, no marks]
		table {%
			0.25 0.74
			0.5 0.74
			0.75 0.74
		};
		\addlegendentry{Local GD}
		\addplot+[very thick, color0, mark=*, mark size=3, mark options={solid,fill opacity=.6}, forget plot]
		table {%
			0.25 0.76
			0.5 0.78
			0.75 0.78
		};
		\addlegendimage{very thick, color0, mark=*, mark options={solid}, mark size=3}
		\addlegendentry{FJ-DGD-1}
		\addplot+[very thick, color1, mark=square*, mark size=3, mark options={solid,fill opacity=.6}, forget plot]
		table {%
			0.25 0.74
			0.5 0.78
			0.75 0.78
		};
		\addlegendimage{very thick, color1, mark=square*, mark options={solid}, mark size=3}
		\addlegendentry{FJ-DGD-2}
	\end{axis}
	
\end{tikzpicture}
	\caption{Minimal local test accuracy achieved for the task in \cref{ex:bin-class}.}
	\label{fig:fj-dgd-loc-test-min}
\end{figure}
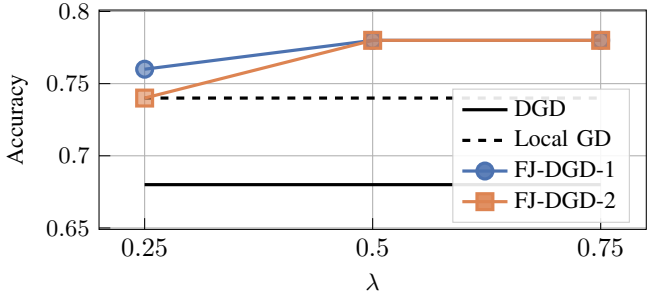

We use the following distributed learning task to conveniently model heterogeneous agents and assess personalization ability.

\begin{ex}[Binary classification]\label{ex:bin-class}
	Consider a classification problem with local loss function
	\begin{equation}\label{eq:log-reg}
		f_i(x_i) = \frac{1}{|\idxset{i}|} \sum_{j\in\idxset{i}} \log\Big( 1 + \exp\left( - \lab{i}{j} \feat{i}{j}^\top x_i \right)\!\Big) + \gamma \norm{x_i}
	\end{equation}
	where $\datapoint{i}{j} = \left( \feat{i}{j}, \lab{i}{j} \right)$, with $\feat{i}{j} \in \Real{n}$ the feature vector and $\lab{i}{j} \in \{ -1, 1 \}$ the classification label.
	Problem~\eqref{eq:distr-learning} with the costs~\eqref{eq:log-reg} is convex.
	We set the size of model parameters $n = 15$.
	Each agent $i$ has access to $|\data[i]| = 450$ training samples and $|\data[i]^{\textnormal{test}}|=50$ test samples such that $\data[i]^{\textnormal{test}}\cap\data[i]=\emptyset$.
	\blue{Following~\cite{Li20mls-federatedOptimization},
	we simulate heterogeneous agents by generating data points $(\feat{i}{j},\lab{i}{j})$ according to $\lab{i}{j}=\argmax(\softmax(W_i\feat{i}{j}+c_i))$ where each element of $W_i$ and $c_i$ is drawn from $\calN(\mu_i,1)$ with $\mu_i\sim\calN(0,1)$,
	and $\feat{i}{j}\sim\calN(\nu_i,\Sigma_i)$ where each element of $\nu_i$ is drawn from $\calN(B_i,1)$ with $B_i\sim\calN(0,1)$ and $\Sigma_i$ is diagonal with $[\Sigma_i]_{jj} = j^{-1.2}$.}
	Agent $i$ classifies the feature vector $\feat{i}{j}$ as
	\begin{equation}
		\hat{\lab{}{}}_{x_i}(\feat{i}{j}) = \begin{cases}
			1 & \text{if $\left(1 + \exp\left( - \feat{i}{j}^\top x_i\right) \right)^{-1} > 0.5$} \\
			- 1 & \text{otherwise}.
		\end{cases}
	\end{equation}
	We set $N = 10$ agents connected through the ring topology.
    \qed
\end{ex}

We train with the setup in \cref{ex:bin-class} and $\gamma = 0.01$.
We choose $\Lambda=\lambda I$ to more easily interpret the experiment through the scalar parameter $\lambda\in[0,1]$.
The results are reported in \cref{fig:fj-dgd,fig:fj-dgd-loc-test-min},
where the label ``Local GD'' refers to each agent independently training on its own dataset without any collaboration.
\blue{Note that Local GD is the instantiation of~\eqref{eq:fj-dgd} and~\eqref{eq:fj-dgd-2} with $\lambda=1$,
whereas~\eqref{eq:dgd} corresponds to setting $\lambda=0$.}

The test accuracy is very similar to the one obtained from training data,
suggesting that the learned models generalize well.
Compared to~\eqref{eq:dgd},
both FJ-inspired modifications exhibit superior personalization performance.
In particular,
\cref{fig:fj-djd-loc-train,fig:fj-djd-loc-test} reveal a graceful improvement in personalization,
along with a controlled degradation of global accuracy in \cref{fig:fj-djd-glob-train,fig:fj-djd-glob-test},
by suitably tuning the parameter $\lambda$.
\blue{As this increases,
the agents turn more ``stubborn'' and consistently bias the local models towards their respective local optimizers,
improving local accuracy.}
Notably,
both~\eqref{eq:fj-dgd} and~\eqref{eq:fj-dgd-2} with $\lambda=0.75$ slightly outperform ``Local GD'' on the local test accuracy (mean accuracy $0.93$ vs. $0.926$),
suggesting that collaboration with other agents is helpful even for personalization purposes,
possibly because agents with similar distributions benefit from each other.
On the other hand,
the global accuracy is much higher for~\eqref{eq:fj-dgd}--\eqref{eq:fj-dgd-2} than ``Local GD'' which highlights the need for (partial) collaboration for generalizing to other agents' distributions.

\autoref{fig:fj-dgd-loc-test-min} spotlights the lowest local accuracy across all agents with the four compared strategies on the test sets.
Note that~\eqref{eq:dgd} and ``Local GD'' feature straight lines because they do not depend on $\lambda$.
DGD performs the worst,
as expected since it values the contribution of all agents equal.
Less intuitive is the fairly poor performance of ``Local GD'' that is even worse than the least personalized FJ-based algorithm.
This also hints at partial collaboration as useful to achieve consistent effective personalization for all agents.

Overall, the best performance is provided by~\eqref{eq:fj-dgd-2},
which scores just slightly higher than~\eqref{eq:fj-dgd} in terms of local accuracy --- and hence provides similar personalization --- but feature significantly higher ($6\%$ to $16\%$ higher mean) global accuracy.
This suggests that a two-stage cascade comprising computation of local and global models and subsequent convex combination of the two is an effective strategy to personalize the local models without excessively compromising global performance.
On the other hand,
\eqref{eq:fj-dgd} personalizes the models as well and requires only two thirds of the memory used by~\eqref{eq:fj-dgd-2},
which is especially useful to train large models as compared to the storage capacity of the agents.

We have demonstrated the superiority of FJ-inspired algorithms for personalization.
We now explore how the algorithms under study behave as the inter-agent heterogeneity varies.

\subsubsection*{Increasing heterogeneity}\label{sec:personalization-heterogeneity}

We propose a study to isolate the effect heterogeneity plays in trading personalization for global accuracy and to assess how well the different algorithms personalize the local models.
To this aim,
we consider the following simplified version of \cref{ex:bin-class} that allows us to easily tune the heterogeneity among agents' local distributions and to visually compare the learned models.

\begin{ex}[Binary classification with 2D features]\label{ex:bin-class-2d}
	The distribution $\calP_i$ of agent $i$ produces samples $\datapoint{i}{j} = \left( \feat{i}{j}, \lab{i}{j} \right)$ where $\feat{i}{j} \in \Real{2}$ and the corresponding label is generated according to the linear model
	\begin{equation}
		\lab{i}{j} = \begin{cases}
			-1	& \mbox{if } w_i^\top \feat{i}{j} + v_{i,j} \ge 0\\
			1	&\mbox{otherwise},
		\end{cases}
	\end{equation}
	with noise $v_{i,j} \sim \calN(0,0.01)$.
	Given a parameter $\theta>0$,
	we construct the ground-truth vectors $\{w_i\}_{i\in\calV}$ as $w_i = [1 \; \theta_i]^\top$
	where the parameters $\{\theta_i\}_{i\in\calV}$ are evenly spaced between $-\theta$ and $\theta$ (included).
	In words,
	a larger value of $\theta$ amplifies the differences among the slopes of vectors $w_i$ and makes the local agents' distributions more heterogeneous.
	Two cases are depicted in \autoref{fig:het-model} with $\theta\in\{0.1,1\}$.
	Agent $i$ trains its parameter $x_i$ to learn $w_i$ using the same loss of \cref{ex:bin-class}.
	The $N=10$ agents communicate according to a circulant graph where each agent has four neighbors.
    We set $\gamma = 10^{-5}$ and train for $1000$ iterations with $500$ training samples per agent.
    \qed
\end{ex}

\begin{figure}
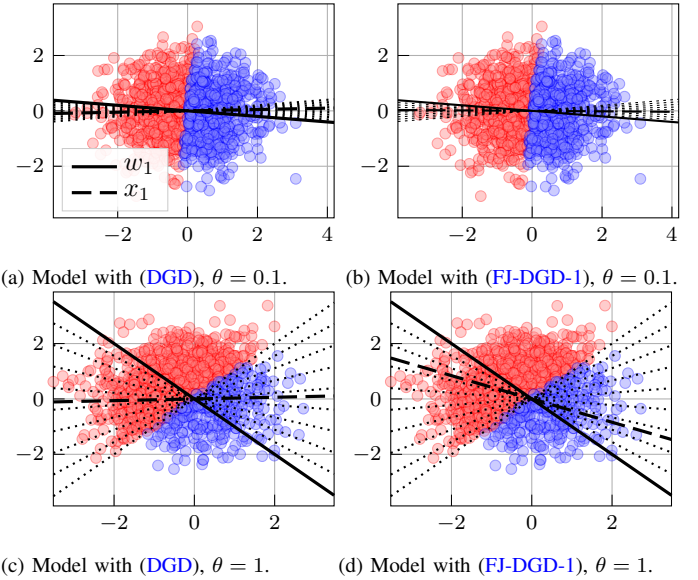

	\centering
	\subfloat[Model with~\eqref{eq:dgd}, $\theta=0.1$.
	\label{fig:het-model-dgd-01}]{\input{figures/7-heterogeneity_model_DGD_0.1}}
	\subfloat[Model with~\eqref{eq:fj-dgd}, $\theta=0.1$.
	\label{fig:het-model-fj-dgd-01}]{\input{figures/7-heterogeneity_model_FJ-DGD-1_0.1}}\\
	\subfloat[Model with~\eqref{eq:dgd}, $\theta=1$.
	\label{fig:het-model-dgd-1}]{\input{figures/7-heterogeneity_model_DGD_1}}
	\subfloat[Model with~\eqref{eq:fj-dgd}, $\theta=1$.
	\label{fig:het-model-fj-dgd-1}]{\input{figures/7-heterogeneity_model_FJ-DGD-1_1}}
	\caption{Samples of and model learned by agent $1$ along with true classifiers $w_i$ of \cref{ex:bin-class-2d}.
		Vector $w_1$ is solid,
		the other vectors $w_i, i\neq1$ dotted.
	}
	\label{fig:het-model}
\end{figure}

\begin{figure}
	\centering
	\subfloat[Local accuracy.
	\label{fig:fj-dgd-het-accuracy-local}]{
\begin{tikzpicture}
	
	\definecolor{color0}{rgb}{0.298039215686275,0.447058823529412,0.690196078431373}
	\definecolor{color1}{rgb}{0.866666666666667,0.517647058823529,0.32156862745098}
	\definecolor{color2}{rgb}{0.333333333333333,0.658823529411765,0.407843137254902}
	\definecolor{color3}{rgb}{0.768627450980392,0.305882352941176,0.32156862745098}
	\definecolor{color4}{rgb}{0.505882352941176,0.447058823529412,0.701960784313725}
		
	\definecolor{darkgray176}{RGB}{176,176,176}
	\definecolor{green01270}{RGB}{0,127,0}
	\definecolor{lightgray204}{RGB}{204,204,204}
	
	\begin{axis}[
		width=\linewidth,
		height=.5\linewidth,
		legend cell align={left},
		legend style={
			fill opacity=0.8,
			draw opacity=1,
			text opacity=1,
			at={(0.03,0.03)},
			anchor=south west,
			draw=white!80!black,
		},
		legend columns=2,
		tick align=inside,
		tick pos=left,
		xtick={.1,.5,1},
		x grid style={darkgray176},
		xmajorgrids,
		xmin=0.055, xmax=1.045,
		xtick style={color=color2},
		xmajorticks=false,
		y grid style={darkgray176},
		ymajorgrids,
            ylabel={Accuracy},
		ymin=0.82, ymax=0.986906666666667,
		ytick style={color=color2},
		]
		\addplot [very thick, color0]
		table {%
			0.1 0.980511111111111
			0.5 0.914377777777778
			1 0.8526
		};
		\addlegendentry{DGD}
		\addplot [color0, mark=*, mark size=1.72504750337168, mark options={solid}, forget plot]
		table {%
			0.1 0.980511111111111
		};
		\addplot [color0, mark=*, mark size=7.4448199739446, mark options={solid}, forget plot]
		table {%
			0.5 0.914377777777778
		};
		\addplot [color0, mark=*, mark size=11.305279494309, mark options={solid}, forget plot]
		table {%
			1 0.8526
		};
		\addplot [very thick, color1, densely dashed]
		table {%
			0.1 0.968044444444444
			0.5 0.9642
			1 0.953866666666667
		};
		\addlegendentry{Local GD}
		\addplot [semithick, color1, mark=*, mark size=6.2679041331391, mark options={solid}, forget plot]
		table {%
			0.1 0.968044444444444
		};
		\addplot [semithick, color1, mark=*, mark size=6.38035962065532, mark options={solid}, forget plot]
		table {%
			0.5 0.9642
		};
		\addplot [semithick, color1, mark=*, mark size=7.17135040746628, mark options={solid}, forget plot]
		table {%
			1 0.953866666666667
		};
		\addplot [very thick, color2, dashdotted]
		table {%
			0.1 0.969466666666667
			0.5 0.947155555555556
			1 0.925222222222222
		};
		\addlegendentry{FJ-DGD-1}
		\addplot [color2, mark=*, mark size=3.83452735027409, mark options={solid, fill opacity=.8}, forget plot]
		table {%
			0.1 0.969466666666667
		};
		\addplot [semithick, color2, mark=*, mark size=4.45323353181583, mark options={solid}, forget plot]
		table {%
			0.5 0.947155555555556
		};
		\addplot [semithick, color2, mark=*, mark size=6.16266987595474, mark options={solid}, forget plot]
		table {%
			1 0.925222222222222
		};
		\addplot [very thick, color3, dotted]
		table {%
			0.1 0.974822222222222
			0.5 0.947111111111111
			1 0.925244444444444
		};
		\addlegendentry{FJ-DGD-2}
		\addplot [thick, color3, mark=*, mark size=3.89686626235321, mark options={solid, fill opacity=.8}, forget plot]
		table {%
			0.1 0.974822222222222
		};
		\addplot [semithick, color3, mark=*, mark size=4.89768652869224, mark options={solid, fill opacity=.7}, forget plot]
		table {%
			0.5 0.947111111111111
		};
		\addplot [semithick, color3, mark=*, mark size=6.34681372378648, mark options={solid, fill opacity=.7}, forget plot]
		table {%
			1 0.925244444444444
		};
	\end{axis}
	
\end{tikzpicture}}\\
	\subfloat[Global accuracy.
	\label{fig:fj-dgd-het-accuracy-global}]{
\begin{tikzpicture}
	\definecolor{color0}{rgb}{0.298039215686275,0.447058823529412,0.690196078431373}
	\definecolor{color1}{rgb}{0.866666666666667,0.517647058823529,0.32156862745098}
	\definecolor{color2}{rgb}{0.333333333333333,0.658823529411765,0.407843137254902}
	\definecolor{color3}{rgb}{0.768627450980392,0.305882352941176,0.32156862745098}
	\definecolor{color4}{rgb}{0.505882352941176,0.447058823529412,0.701960784313725}
	
	\definecolor{darkgray176}{RGB}{176,176,176}
	\definecolor{lightgray204}{RGB}{204,204,204}
	
	\begin{axis}[
		width=\linewidth,
		height=.5\linewidth,
		tick align=inside,
		tick pos=left,
		x grid style={darkgray176},
		xmajorgrids,
		xmin=0.055, xmax=1.045,
		xtick={.1,.5,1},
		xlabel={$\theta$},
		y grid style={darkgray176},
		ymajorgrids,
		ylabel={Accuracy},
		ymin=0.805207333333333, ymax=0.988672666666667,
		ytick style={color=color0}
		]
		\addplot [very thick, color0]
		table {%
			0.1 0.980333333333334
			0.5 0.9137
			1 0.85084
		};
		\addplot [semithick, color0, mark=*, mark size=0.00447213595500817, mark options={solid}, forget plot]
		table {%
			0.1 0.980333333333334
		};
		\addplot [semithick, color0, mark=*, mark size=0.0109797793946653, mark options={solid}, forget plot]
		table {%
			0.5 0.9137
		};
		\addplot [semithick, color0, mark=*, mark size=0.0360493488922586, mark options={solid}, forget plot]
		table {%
			1 0.85084
		};
		\addplot [very thick, color1, densely dashed]
		table {%
			0.1 0.95724
			0.5 0.889133333333333
			1 0.813546666666667
		};
		\addplot [semithick, color1, mark=*, mark size=5.08639488133678, mark options={solid}, forget plot]
		table {%
			0.1 0.95724
		};
		\addplot [semithick, color1, mark=*, mark size=2.67123982866051, mark options={solid}, forget plot]
		table {%
			0.5 0.889133333333333
		};
		\addplot [semithick, color1, mark=*, mark size=4.28364336102394, mark options={solid}, forget plot]
		table {%
			1 0.813546666666667
		};
		\addplot [very thick, color2, dashdotted]
		table {%
			0.1 0.967951111111111
			0.5 0.905373333333333
			1 0.834997777777778
		};
		\addplot [semithick, color2, mark=*, mark size=3.44315546491232, mark options={solid}, forget plot]
		table {%
			0.1 0.967951111111111
		};
		\addplot [semithick, color2, mark=*, mark size=0.83799071859087, mark options={solid}, forget plot]
		table {%
			0.5 0.905373333333333
		};
		\addplot [semithick, color2, mark=*, mark size=1.69001213013398, mark options={solid}, forget plot]
		table {%
			1 0.834997777777778
		};
		\addplot [very thick, color3, dotted]
		table {%
			0.1 0.96904
			0.5 0.905035555555556
			1 0.834704444444444
		};
		\addplot [semithick, color3, mark=*, mark size=3.21591106842213, mark options={solid}, forget plot]
		table {%
			0.1 0.96904
		};
		\addplot [semithick, color3, mark=*, mark size=1.01462987230703, mark options={solid}, forget plot]
		table {%
			0.5 0.905035555555556
		};
		\addplot [semithick, color3, mark=*, mark size=1.94359177812626, mark options={solid}, forget plot]
		table {%
			1 0.834704444444444
		};
	\end{axis}
	
\end{tikzpicture}}
	\caption{Test accuracy with DGD and its FJ-based variations under increasing inter-agent heterogeneity for the binary classification task in \cref{ex:bin-class-2d}.}
	\label{fig:fj-dgd-het-accuracy}
\end{figure}
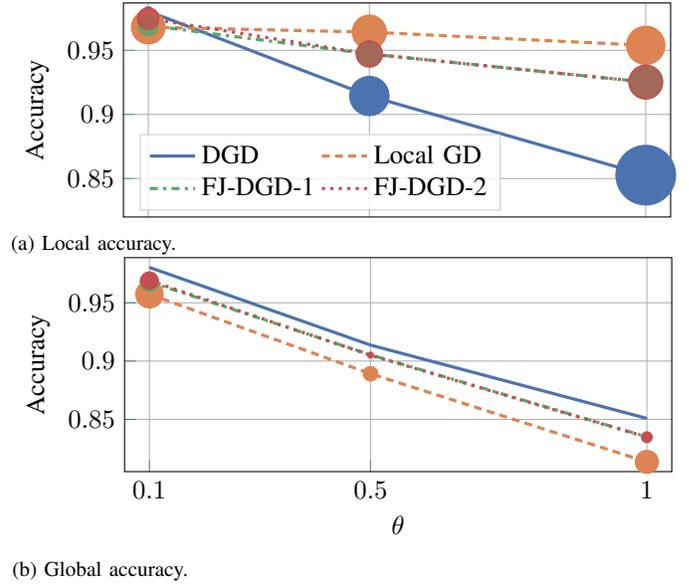

We set $\lambda=0.5$ for both~\eqref{eq:fj-dgd} and~\eqref{eq:fj-dgd-2}.
\autoref{fig:fj-dgd-het-accuracy} summarizes the performance achieved with the four algorithms previously compared.
Each mark displays the average accuracy as the location on the $y$-axis and the standard deviation across agents as the size.\footnote{
    We scale all standard deviations by $100$ to make the marks visible.
}
As $\theta$ increases,
the local distributions become more and more different and~\eqref{eq:dgd} struggles to balance global for local accuracy.
The latter degrades both in average and deviation,
with some agents hitting low scores.
If each agent independently trains (``Local GD''),
the local accuracy barely changes with $\theta$ but the global accuracy decreases.
On the other hand,
the FJ-based variants of DGD gracefully mediate between global accuracy and personalization,
achieving significantly higher local accuracy than~\eqref{eq:dgd} with smaller standard deviation.
Interestingly,
both~\eqref{eq:fj-dgd} and~\eqref{eq:fj-dgd-2} perform almost identical for this task,
especially with high heterogeneity (large $\theta$).

A visual explanation of enhanced personalization is given in \autoref{fig:het-model} focusing on agent $1$.
When $\theta$ is small,
all ground-truth local classifiers are similar to each other,
and both~\eqref{eq:dgd} and~\eqref{eq:fj-dgd} yield similar models $x_1$.
However,
when $\theta$ is large and the local models differ significantly,
the model $x_1$ learned by~\eqref{eq:fj-dgd} is closer to the true vector $w_1$ compared to the one learned by~\eqref{eq:dgd},
enhancing personalization while retaining high global accuracy.
This behavior can be easily tuned through the parameter $\lambda$ in~\eqref{eq:fj-dgd} or~\eqref{eq:fj-dgd-2} and makes these algorithms flexible to various needs.

The experiments of \cref{ex:bin-class,ex:bin-class-2d} suggest that the FJ-inspired adaptations of DGD are effective for learning personalized models while also retaining good global accuracy according to the nominal problem~\eqref{eq:distr-learning}.
Their simplicity makes them attractive both for practical implementation and for interpretation of the models,
whose degree of personalization can be easily tuned through the parameter $\lambda$.
In the next section,
we consider the scenario with malicious agents.

\subsection{Evaluating Resilience}\label{sec:resilience-experiments}

For this set of experiments,
we consider the multi-class classification task described next.

\begin{ex}[Image classification]\label{ex:task-mnist}
	We study classification on the MNIST dataset with samples $d = (\feat{}{},\lab{}{})$ where $\feat{}{}$ is the flattened (one-dimensional) photo of a handwritten digit and $\lab{}{}$ is the corresponding numerical value from $0$ to $9$.
	Each agent $i$ learns the parameter $x_i\in\Real{p\times c}$ of a multi-class logistic classifier,
	where $p = 784$ is the image size ($28\times28$ pixels) and $c = 10$ is the number of classes (digits).
	Denoting the $\ell$th column of $x_i$ by $[x_i]_\ell$,
	the multi-class logistic loss is
	\begin{multline}\label{eq:multiclass-logistic-loss}
		\hspace{-3.5mm}
		f_i(x_i) = \frac{1}{|\idxset{i}|} \sum_{j\in\idxset{i}} \lr \log\lr\sum_{c'=1}^c \exp\lr\feat{i}{j}^\top [x_i]_{c'}\rr\rr - \feat{i}{j}^\top [x_i]_{\lab{i}{j}} \rr\\
		+ \gamma \norm{x_i}
	\end{multline}
	with regularization weight $\gamma = 0.1$.
	Labels are then assigned according to the softmax policy as
	\begin{equation}
		\hat{\lab{}{}} = \argmax_{\ell\in\{1,\dots,c\}} \dfrac{\exp(\feat{}{}^\top [x_i]_{\ell})}{\sum_{c'=1}^c \exp(\feat{}{}^\top [x_i]_{c'})} - 1.
	\end{equation}
	We simulate $N = 100$ total agents communicating over a random geometric network on the unitary square in $\Real{2}$ with communication radius $\rho = 0.25$.
	We generate three datasets that make collaborative learning increasingly difficult.
	\begin{description}
		\item[Homogeneous (Hom):] We randomly assign $554$ samples of MNIST to each agent.
		Samples are randomly split between local training set ($443$ samples) and test set ($111$ samples) for each agent.
		In this case,
		all agents most likely have representatives of all classes in both training and test data.
		\item[\boldmath Heterogeneous-$2$ (Het-$2$):] We modify the ``Hom'' datasets and randomly remove two classes from each agent's local train and test data.
		Thus,
		each agent's local datasets contain at most eight out of the ten digit classes.
		\item[\boldmath Heterogeneous-$5$ (Het-$5$):] We randomly remove five classes from each agent's local ``Hom'' train and test data.
	\end{description}
	In all experiments,
	we show only the global accuracy since in this case we are not interested in personalized models.
    \qed
\end{ex}

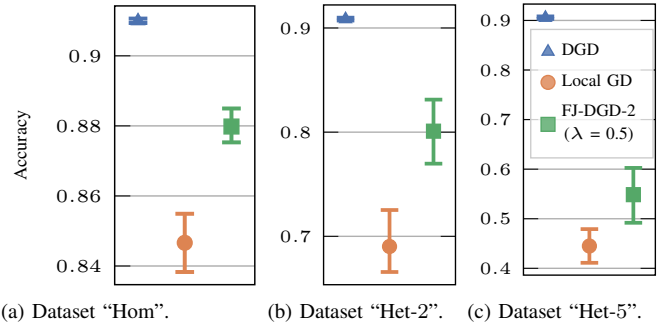
\begin{figure}
	\centering
	\subfloat[Dataset ``Hom''.
	\label{fig:mnist-no-mal-hom}]{\resizebox{.4\linewidth}{!}{

\begin{tikzpicture}
	
	\definecolor{darkgray176}{RGB}{176,176,176}
	\definecolor{lightgray204}{RGB}{204,204,204}
	\definecolor{steelblue76114176}{RGB}{76,114,176}
	\definecolor{color1}{rgb}{0.866666666666667,0.517647058823529,0.32156862745098}
	\definecolor{color2}{rgb}{0.333333333333333,0.658823529411765,0.407843137254902}
	
	\begin{axis}[
		width=.33\linewidth,
		height=.48\linewidth,
		legend style={fill opacity=0.8, draw opacity=1, text opacity=1, draw=lightgray204},
		tick align=inside,
		tick pos=left,
		x grid style={darkgray176},
		xmin=-0.5, xmax=2.5,
		xtick style={color=black},
		xtick={0,1,2},
		xticklabels={DGD,Local GD,FJ-DGD-1},
		xmajorticks=false,
		y grid style={darkgray176},
		yticklabel style={font=\tiny},
		ylabel={Accuracy},
		ylabel style={font=\tiny},
		ymajorgrids,
		ymin=0.83467119038105, ymax=0.914247760176659,
		ytick style={color=black}
		]
		\addplot [draw=steelblue76114176, fill=steelblue76114176, forget plot, mark=triangle*, only marks]
		table{%
			0 0.909965753555298
		};
		\addplot [line width=1.08pt, steelblue76114176, forget plot]
		table {%
			0 0.909369349479675
			0 0.910630643367767
		};
		\addplot [line width=1.08pt, steelblue76114176, forget plot]
		table {%
			-0.2 0.909369349479675
			0.2 0.909369349479675
		};
		\addplot [line width=1.08pt, steelblue76114176, forget plot]
		table {%
			-0.2 0.910630643367767
			0.2 0.910630643367767
		};
		\addplot [draw=color1, fill=color1, forget plot, mark=*, only marks]
		table{%
			1 0.846627950668335
		};
		\addplot [line width=1.08pt, color1, forget plot]
		table {%
			1 0.838288307189941
			1 0.854909904301167
		};
		\addplot [line width=1.08pt, color1, forget plot]
		table {%
			0.8 0.838288307189941
			1.2 0.838288307189941
		};
		\addplot [line width=1.08pt, color1, forget plot]
		table {%
			0.8 0.854909904301167
			1.2 0.854909904301167
		};
		\addplot [draw=color2, fill=color2, forget plot, mark=square*, only marks]
		table{%
			2 0.879817068576813
		};
		\addplot [line width=1.08pt, color2, forget plot]
		table {%
			2 0.87529281526804
			2 0.88496620208025
		};
		\addplot [line width=1.08pt, color2, forget plot]
		table {%
			1.8 0.87529281526804
			2.2 0.87529281526804
		};
		\addplot [line width=1.08pt, color2, forget plot]
		table {%
			1.8 0.88496620208025
			2.2 0.88496620208025
		};
	\end{axis}
	
\end{tikzpicture}}}%
	\subfloat[Dataset ``Het-$2$''.
	\label{fig:mnist-no-mal-het2}]{\resizebox{.3\linewidth}{!}{

\begin{tikzpicture}
	
	\definecolor{darkgray176}{RGB}{176,176,176}
	\definecolor{lightgray204}{RGB}{204,204,204}
	\definecolor{steelblue76114176}{RGB}{76,114,176}
	\definecolor{color1}{rgb}{0.866666666666667,0.517647058823529,0.32156862745098}
	\definecolor{color2}{rgb}{0.333333333333333,0.658823529411765,0.407843137254902}
	
	\begin{axis}[
		width=.33\linewidth,
		height=.5\linewidth,
		unbounded coords=jump,
		x grid style={darkgray176},
		xmin=-0.5, xmax=2.5,
		xmajorticks=false,
		y grid style={darkgray176},
		yticklabel style={font=\tiny},
		ymajorgrids,
		ymin=0.653558563068509, ymax=0.922117110714316,
		ytick style={color=black}
		]
		\addplot [draw=steelblue76114176, fill=steelblue76114176, forget plot, mark=triangle*, only marks]
		table{%
			0 0.908458530902863
		};
		\addplot [line width=1.08pt, steelblue76114176, forget plot]
		table {%
			0 0.907117128372192
			0 0.909909904003143
		};
		\addplot [line width=1.08pt, steelblue76114176, forget plot]
		table {%
			-0.2 0.907117128372192
			0.2 0.907117128372192
		};
		\addplot [line width=1.08pt, steelblue76114176, forget plot]
		table {%
			-0.2 0.909909904003143
			0.2 0.909909904003143
		};
		\addplot [draw=color1, fill=color1, forget plot, mark=*, only marks]
		table{%
			1 0.690239667892456
		};
		\addplot [line width=1.08pt, color1, forget plot]
		table {%
			1 0.665765769779682
			1 0.725258998572826
		};
		\addplot [line width=1.08pt, color1, forget plot]
		table {%
			0.8 0.665765769779682
			1.2 0.665765769779682
		};
		\addplot [line width=1.08pt, color1, forget plot]
		table {%
			0.8 0.725258998572826
			1.2 0.725258998572826
		};
		\addplot [line width=1.08pt, color2, forget plot]
		table {%
			2 0.800943195819855
		};
		\addplot [line width=1.08pt, color2, forget plot]
		table {%
			2 0.769752264022827
			2 0.831238746643066
		};
		\addplot [line width=1.08pt, color2, forget plot]
		table {%
			1.8 0.769752264022827
			2.2 0.769752264022827
		};
		\addplot [line width=1.08pt, color2, forget plot]
		table {%
			1.8 0.831238746643066
			2.2 0.831238746643066
		};
		\addplot [draw=color2, fill=color2, mark=square*, only marks]
		table{%
			2 0.800943195819855
		};
	\end{axis}
	
\end{tikzpicture}}}%
	\subfloat[Dataset ``Het-$5$''.
	\label{fig:mnist-no-mal-het5}]{\resizebox{.3\linewidth}{!}{

\begin{tikzpicture}
	
	\definecolor{darkgray176}{RGB}{176,176,176}
	\definecolor{lightgray204}{RGB}{204,204,204}
	\definecolor{steelblue76114176}{RGB}{76,114,176}
	\definecolor{color1}{rgb}{0.866666666666667,0.517647058823529,0.32156862745098}
	\definecolor{color2}{rgb}{0.333333333333333,0.658823529411765,0.407843137254902}
	
	\begin{axis}[
		width=.33\linewidth,
		height=.49\linewidth,
		legend cell align={left},
		legend style={
			fill opacity=0.8,
			draw opacity=1,
			text opacity=1,
			at={(0.05,.43)},
			anchor=south west,
			draw=white!80!black,
			font=\tiny,
		},
		unbounded coords=jump,
		x grid style={darkgray176},
		xmin=-0.5, xmax=2.5,
		xmajorticks=false,
		y grid style={darkgray176},
		yticklabel style={font=\tiny},
		ymajorgrids,
            ytick={0.4,0.5,0.6,0.7,0.8,0.9},
		ymin=0.386274207755923, ymax=0.933016318455339,
		ytick style={color=black}
		]
		\addplot [draw=steelblue76114176, fill=steelblue76114176, mark=triangle*, only marks]
		table{%
			x  y
			0 0.905354917049408
		};
		\addlegendentry{DGD}
		\addplot [line width=1.08pt, steelblue76114176, forget plot]
		table {%
			0 0.902916669845581
			0 0.908164404332638
		};
		\addplot [line width=1.08pt, steelblue76114176, forget plot]
		table {%
			-0.2 0.902916669845581
			0.2 0.902916669845581
		};
		\addplot [line width=1.08pt, steelblue76114176, forget plot]
		table {%
			-0.2 0.908164404332638
			0.2 0.908164404332638
		};
		\addplot [draw=color1, fill=color1, mark=*, only marks]
		table{%
			1 0.445003598928452
		};
		\addlegendentry{Local GD}
		\addplot [line width=1.08pt, color1, forget plot]
		table {%
			1 0.411126121878624
			1 0.479099094867706
		};
		\addplot [line width=1.08pt, color1, forget plot]
		table {%
			0.8 0.411126121878624
			1.2 0.411126121878624
		};
		\addplot [line width=1.08pt, color1, forget plot]
		table {%
			0.8 0.479099094867706
			1.2 0.479099094867706
		};
		\addplot [draw=color2, fill=color2, mark=square*, only marks]
		table{%
			2 0.548558592796326
		};
		\addlegendentry{\shortstack{FJ-DGD-2 \\ ($\lambda$ = 0.5)}}
		\addplot [line width=1.08pt, color2, forget plot]
		table {%
			2 0.491846852004528
			2 0.602567560970783
		};
		\addplot [line width=1.08pt, color2, forget plot]
		table {%
			1.8 0.491846852004528
			2.2 0.491846852004528
		};
		\addplot [line width=1.08pt, color2, forget plot]
		table {%
			1.8 0.602567560970783
			2.2 0.602567560970783
		};
	\end{axis}
	
\end{tikzpicture}}}
	\caption{Accuracy on MNIST dataset in \cref{ex:task-mnist} without malicious agents.
		Marks show the mean and bars the $75\%$ percentile interval across agents.}
	\label{fig:mnist-no-mal}
\end{figure}

Motivated by the higher performance of~\eqref{eq:fj-dgd-2} in the previous section,
we focus on it for the next experiments.
First,
we set the baseline evaluating classification accuracy in the ideal case with no malicious agents,
running $1000$ learning iterations.
We use $\lambda=0.5$ for this test.
As expected,
the algorithm~\eqref{eq:dgd} performs best,
followed by~\eqref{eq:fj-dgd-2} and lastly by local training,
and the accuracy decreases as more classes are removed from each local dataset.
Note that the steady-state accuracy on the dataset ``Hom'' is expected to be almost equal with the three algorithms because all agents qualitatively have the same information in this case.
Nonetheless,
owing to different speeds of convergence and augmented training data compared to local datasets,
both~\eqref{eq:dgd} and ~\eqref{eq:fj-dgd-2} achieve higher accuracy than ``Local GD'' even after $1000$ learning iterations.
Such a faster convergence can be a further argument in favor of collaborative training.

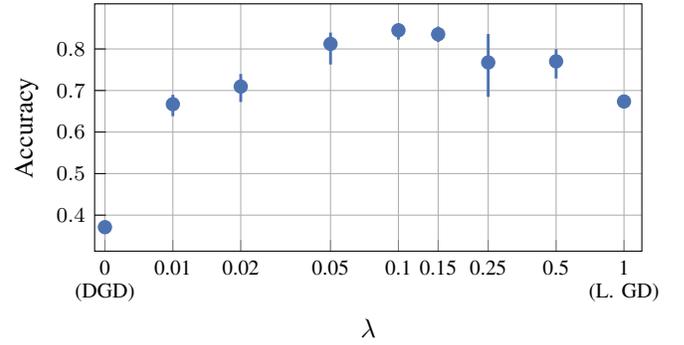
\begin{figure}
	\centering
	\begin{tikzpicture}
	\definecolor{steelblue76114176}{RGB}{76,114,176}
	
	\definecolor{darkgray176}{RGB}{176,176,176}
	\definecolor{lightgray204}{RGB}{204,204,204}
	
	\begin{axis}[
		width=\linewidth,
		height=.55\linewidth,
		tick align=inside,
		tick pos=left,
		x grid style={darkgray176},
		xmajorgrids,
		xtick={0.005, 0.01, 0.02, 0.05, 0.1, 0.15, 0.25, 0.5, 1},
		xticklabels={\shortstack{0\\(DGD)}, 0.01, 0.02, 0.05, 0.1, 0.15, 0.25, 0.5, \shortstack{1\\(L.~GD)}},
		xticklabel style={font=\footnotesize},
		xmode=log,
		xmin=0.0045,
		xmax=1.2,
		log ticks with fixed point,
		xlabel={$\lambda$},
		y grid style={darkgray176},
		ymajorgrids,
		ytick={0.4,0.5,0.6,0.7,0.8},
		yticklabel style={font=\footnotesize},
		ylabel={Accuracy},
		ytick style={color=black}
		]
		\addplot [draw=steelblue76114176, fill=steelblue76114176, forget plot, mark=*, only marks, mark size=2.5]
		table{%
			0.005 0.37104806303978
			0.01 0.667085111141205
			0.02 0.709553599357605
			0.05 0.812166273593903
			0.1 0.84497606754303
			0.15 0.835470497608185
			0.25 0.767553627490997
			0.5 0.770182132720947
			1 0.673615574836731
		};
		\addplot [line width=1.08pt, steelblue76114176, forget plot]
		table {%
			0.005 0.362927928566933
			0.005 0.378704946488142
		};
		\addplot [line width=1.08pt, steelblue76114176, forget plot]
		table {%
			0.01 0.638108126819134
			0.01 0.689797289669514
		};
		\addplot [line width=1.08pt, steelblue76114176, forget plot]
		table {%
			0.02 0.672297284007072
			0.02 0.739831060171127
		};
		\addplot [line width=1.08pt, steelblue76114176, forget plot]
		table {%
			0.05 0.762657657265663
			0.05 0.839448213577271
		};
		\addplot [line width=1.08pt, steelblue76114176, forget plot]
		table {%
			0.1 0.822590090334415
			0.1 0.858907647430897
		};
		\addplot [line width=1.08pt, steelblue76114176, forget plot]
		table {%
			0.15 0.816722996532917
			0.15 0.853727482259274
		};
		\addplot [line width=1.08pt, steelblue76114176, forget plot]
		table {%
			0.25 0.684707187116146
			0.25 0.835867129266262
		};
		\addplot [line width=1.08pt, steelblue76114176, forget plot]
		table {%
			0.5 0.728975214064121
			0.5 0.799448199570179
		};
		\addplot [line width=1.08pt, steelblue76114176, forget plot]
		table {%
			1 0.660743251442909
			1 0.688198186457157
		};
	\end{axis}
	
\end{tikzpicture}
	\caption{Accuracy on MNIST ``Hom'' in \cref{ex:task-mnist} with malicious agents after $300$ learning iterations.
	}
	\label{fig:mnist-mal-hom}
\end{figure}

We then randomly select $10$ malicious agents across the network.
This yields a ratio between malicious and honest agents of over $11\%$.
We set stealthy attacks whereby each malicious agent $m$ trains its local parameter $x_m$ with~\eqref{eq:dgd} but,
at each iteration $k$,
communicates to its neighbors purposely corrupted parameters $\tilde{x}_m$ computed as follows
\begin{equation}\label{eq:attack}
	\tilde{x}_{m,k} = x_{m,k} + v_k, \quad
	v_k \sim \calN(0, \operatorname{diag}\lr\min\{\eta|x_{m,k}|, \kappa\}\rr),
\end{equation}
where the minimization is element-wise.

\subsubsection{Dataset ``Hom''}
We first run a shorter training on dataset ``Homogeneous'' with $\eta=\kappa=5$.
The accuracy is reported in \autoref{fig:mnist-mal-hom} where we make explicit that $\lambda=0$ and $\lambda=1$ are respectively~\eqref{eq:dgd} and ``Local GD'' (''L. GD'').
Our approach tailored to personalization enhances resilience as well,
outperforming the benchmarks.
Algorithm~\eqref{eq:dgd} treats all agents,
malicious ones included,
equally.
``Local GD'' is insensitive to attacks but does not leverage collaboration with honest agents.
Our algorithm~\eqref{eq:fj-dgd-2} with $\lambda \in (0,1)$ stands in the middle;
agents do not fully rely on the others but retain benefits of mutual collaboration.
The U-shaped curve confirms the intuition that partially reducing collaboration is beneficial as it mitigates the influence of malicious agents,
but excessively doing so degrades performance because updates are too conservative.
The same behavior was analytically characterized in~\cite{Ballotta24tac-competitionCollaboration} for resilient consensus.
The rapid rise of the accuracy suggests that~\eqref{eq:fj-dgd-2} is sensitive to small values of $\lambda$,
which already improve resilience significantly.

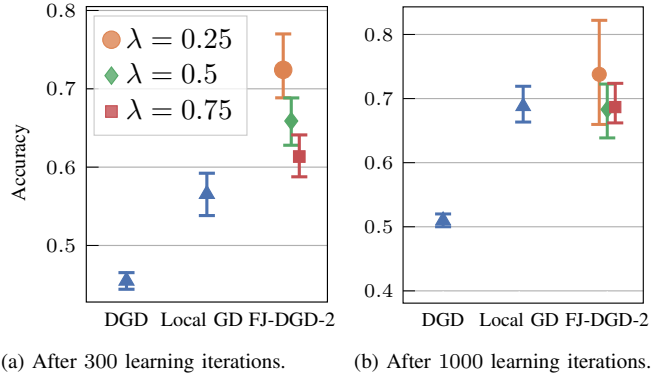
\begin{figure}
	\centering
	\subfloat[After $300$ learning iterations.
	\label{fig:mnist-mal-het2-300}]{\resizebox{.53\linewidth}{!}{

\begin{tikzpicture}
	
	\definecolor{darkgray176}{RGB}{176,176,176}
	\definecolor{indianred1967882}{RGB}{196,78,82}
	\definecolor{lightgray204}{RGB}{204,204,204}
	\definecolor{mediumseagreen85168104}{RGB}{85,168,104}
	\definecolor{peru22113282}{RGB}{221,132,82}
	\definecolor{steelblue76114176}{RGB}{76,114,176}

	\node[rotate=90,font=\scriptsize] at (-.8,1.7) {Accuracy};
	
	\begin{axis}[
		height=.585\linewidth,
		width=.51\linewidth,
		legend cell align={left},
		legend style={
			fill opacity=0.8,
			draw opacity=1,
			text opacity=1,
			draw=lightgray204
		},
		legend pos=north west,
		tick align=inside,
		tick pos=left,
		unbounded coords=jump,
		x grid style={darkgray176},
		xmin=-0.5, xmax=2.5,
		xtick style={color=black},
		xtick={0,.94,2.06},
		xticklabels={DGD,Local GD,FJ-DGD-2},
		xtick style={color=white},
		xticklabel style={font=\scriptsize},
		yticklabel style={font=\scriptsize},
		y grid style={darkgray176},
		ymajorgrids,
		ymin=0.427786007151008, ymax=0.805,
		ytick style={color=black}
		]
		\addplot [line width=1.08pt, steelblue76114176, forget plot]
		table {%
			0 0.444084465503693
			0 0.465311139822006
		};
		\addplot [line width=1.08pt, steelblue76114176, forget plot]
		table {%
			0.1 0.444084465503693
			-0.1 0.444084465503693
		};
		\addplot [line width=1.08pt, steelblue76114176, forget plot]
		table {%
			0.1 0.465311139822006
			-0.1 0.465311139822006
		};
		\addplot [line width=1.08pt, steelblue76114176, forget plot]
		table {%
			1 0.538166128098965
			1 0.592224359512329
		};
		\addplot [line width=1.08pt, steelblue76114176, forget plot]
		table {%
			1.1 0.538166128098965
			0.9 0.538166128098965
		};
		\addplot [line width=1.08pt, steelblue76114176, forget plot]
		table {%
			1.1 0.592224359512329
			0.9 0.592224359512329
		};
		\addplot [draw=steelblue76114176, fill=steelblue76114176, forget plot, mark=triangle*, only marks, mark size=3]
		table{%
			0 0.45489764213562
			1 0.565600335597992
		};
		\addplot [line width=1.08pt, peru22113282, forget plot]
		table {%
			1.95 0.724182307720184
		};
		\addplot [line width=1.08pt, peru22113282, forget plot]
		table {%
			1.95 0.688484534621239
			1.95 0.770053632557392
		};
		\addplot [line width=1.08pt, peru22113282, forget plot]
		table {%
			1.85 0.688484534621239
			2.05 0.688484534621239
		};
		\addplot [line width=1.08pt, peru22113282, forget plot]
		table {%
			1.85 0.770053632557392
			2.05 0.770053632557392
		};
		\addplot [draw=peru22113282, fill=peru22113282, mark=*, only marks, mark size=3]
		table{%
			1.95 0.724182307720184
		};
		\addlegendentry{$\lambda = 0.25$}
		\addplot [line width=1.08pt, mediumseagreen85168104, forget plot]
		table {%
			2.05 0.658888518810272
		};
		\addplot [line width=1.08pt, mediumseagreen85168104, forget plot]
		table {%
			2.05 0.627946600317955
			2.05 0.688358835875988
		};
		\addplot [line width=1.08pt, mediumseagreen85168104, forget plot]
		table {%
			1.95 0.627946600317955
			2.15 0.627946600317955
		};
		\addplot [line width=1.08pt, mediumseagreen85168104, forget plot]
		table {%
			1.95 0.688358835875988
			2.15 0.688358835875988
		};
		\addplot [draw=mediumseagreen85168104, fill=mediumseagreen85168104, mark=diamond*, only marks, mark size=3]
		table{%
			2.05 0.658888518810272
		};
		\addlegendentry{$\lambda = 0.5$}
		\addplot [line width=1.08pt, indianred1967882, forget plot]
		table {%
			2.15 0.613519310951233
		};
		\addplot [line width=1.08pt, indianred1967882, forget plot]
		table {%
			2.15 0.587657824158669
			2.15 0.641129463911057
		};
		\addplot [line width=1.08pt, indianred1967882, forget plot]
		table {%
			2.05 0.587657824158669
			2.25 0.587657824158669
		};
		\addplot [line width=1.08pt, indianred1967882, forget plot]
		table {%
			2.05 0.641129463911057
			2.25 0.641129463911057
		};
		\addplot [draw=indianred1967882, fill=indianred1967882, mark=square*, only marks, mark size=2]
		table{%
			2.15 0.613519310951233
		};
		\addlegendentry{$\lambda = 0.75$}
	\end{axis}
	
\end{tikzpicture}}}%
	\subfloat[After $1000$ learning iterations.
	\label{fig:mnist-mal-het2-1000}]{\resizebox{.47\linewidth}{!}{

\begin{tikzpicture}
	
	\definecolor{darkgray176}{RGB}{176,176,176}
	\definecolor{indianred1967882}{RGB}{196,78,82}
	\definecolor{lightgray204}{RGB}{204,204,204}
	\definecolor{mediumseagreen85168104}{RGB}{85,168,104}
	\definecolor{peru22113282}{RGB}{221,132,82}
	\definecolor{steelblue76114176}{RGB}{76,114,176}
	
	\begin{axis}[
		height=.6\linewidth,
		width=.52\linewidth,
		legend cell align={left},
		legend style={
			fill opacity=0.8,
			draw opacity=1,
			text opacity=1,
			draw=lightgray204
		},
		legend pos=north west,
		tick align=inside,
		tick pos=left,
		unbounded coords=jump,
		x grid style={darkgray176},
		xmin=-0.5, xmax=2.5,
		xtick style={color=black},
		xtick={0,.95,2.05},
		xticklabels={DGD,Local GD,FJ-DGD-2},
		xtick style={color=white},
		xticklabel style={font=\scriptsize},
		xlabel style={font=\scriptsize},
		y grid style={darkgray176},
		ytick={.4,.5,.6,.7,.8},
		yticklabel style={font=\scriptsize},
		ymajorgrids,
		ymin=0.3808598106727, ymax=0.843070395477116,
		ytick style={color=black}
		]
		\addplot [line width=1.08pt, steelblue76114176, forget plot]
		table {%
			0 0.5
			0 0.52
		};
		\addplot [line width=1.08pt, steelblue76114176, forget plot]
		table {%
			0.1 0.5
			-0.1 0.5
		};
		\addplot [line width=1.08pt, steelblue76114176, forget plot]
		table {%
			0.1 0.52
			-0.1 0.52
		};
		\addplot [line width=1.08pt, steelblue76114176, forget plot]
		table {%
			1 0.663367137312889
			1 0.719189181923866
		};
		\addplot [line width=1.08pt, steelblue76114176, forget plot]
		table {%
			1.1 0.663367137312889
			0.9 0.663367137312889
		};
		\addplot [line width=1.08pt, steelblue76114176, forget plot]
		table {%
			1.1 0.719189181923866
			0.9 0.719189181923866
		};
		\addplot [draw=steelblue76114176, fill=steelblue76114176, forget plot, mark=triangle*, only marks, mark size=3]
		table{%
			0 0.51
			1 0.687843799591064
		};
		\addplot [line width=1.08pt, peru22113282, forget plot]
		table {%
			1.95 0.737849771976471
		};
		\addplot [line width=1.08pt, peru22113282, forget plot]
		table {%
			1.95 0.659639626741409
			1.95 0.822060823440552
		};
		\addplot [line width=1.08pt, peru22113282, forget plot]
		table {%
			1.85 0.659639626741409
			2.05 0.659639626741409
		};
		\addplot [line width=1.08pt, peru22113282, forget plot]
		table {%
			1.85 0.822060823440552
			2.05 0.822060823440552
		};
		\addplot [draw=peru22113282, fill=peru22113282, mark=*, only marks, mark size=2.5]
		table{%
			1.95 0.737849771976471
		};
		\addplot [line width=1.08pt, mediumseagreen85168104, forget plot]
		table {%
			2.05 0.683022081851959
		};
		\addplot [line width=1.08pt, mediumseagreen85168104, forget plot]
		table {%
			2.05 0.638626106083393
			2.05 0.722623862326145
		};
		\addplot [line width=1.08pt, mediumseagreen85168104, forget plot]
		table {%
			1.95 0.638626106083393
			2.15 0.638626106083393
		};
		\addplot [line width=1.08pt, mediumseagreen85168104, forget plot]
		table {%
			1.95 0.722623862326145
			2.15 0.722623862326145
		};
		\addplot [draw=mediumseagreen85168104, fill=mediumseagreen85168104, mark=diamond*, only marks, mark size=3]
		table{%
			2.05 0.683022081851959
		};
		\addplot [line width=1.08pt, indianred1967882, forget plot]
		table {%
			2.15 0.686741828918457
		};
		\addplot [line width=1.08pt, indianred1967882, forget plot]
		table {%
			2.15 0.662038303911686
			2.15 0.723761267960072
		};
		\addplot [line width=1.08pt, indianred1967882, forget plot]
		table {%
			2.05 0.662038303911686
			2.25 0.662038303911686
		};
		\addplot [line width=1.08pt, indianred1967882, forget plot]
		table {%
			2.05 0.723761267960072
			2.25 0.723761267960072
		};
		\addplot [draw=indianred1967882, fill=indianred1967882, mark=square*, only marks, mark size=2]
		table{%
			2.15 0.686741828918457
		};
	\end{axis}
	
\end{tikzpicture}}}
	\caption{Accuracy on MNIST ``Het-2'' in \cref{ex:task-mnist} with malicious agents.}
	\label{fig:mnist-mal-het2}
\end{figure}

\begin{figure*}
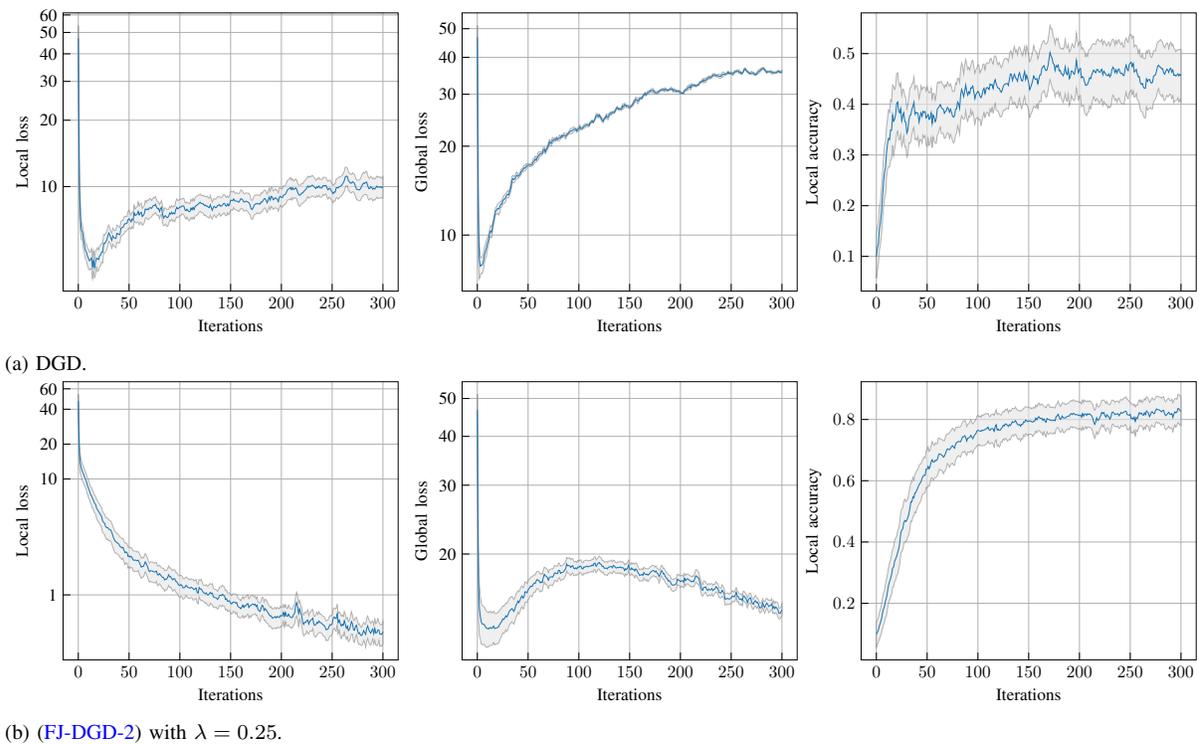

	\centering
	\subfloat[DGD.
	\label{fig:mnist-mal-het2-train-dgd}]{\input{figures/mnist/random_graph/mal_tx-noise/het_2/plot_training_N100_mal10_it300_bzNone_noise3_dgd}}\\
	\subfloat[\eqref{eq:fj-dgd-2} with $\lambda = 0.25$.
	\label{fig:mnist-mal-het2-train-fjdgd}]{\input{figures/mnist/random_graph/mal_tx-noise/het_2/plot_training_N100_mal10_it300_bzNone_noise3_fjdgd0.25}}
	\caption{Local losses, global losses, and local test accuracy for all agents on MNIST ``Het-2'' in \cref{ex:task-mnist} with malicious agents.}
	\label{fig:mnist-mal-het2-train}
\end{figure*}

\subsubsection{Dataset ``Het-2''}
We then train the agents on ``Heterogeneous-2'' with the choice $\eta=10$ and $\kappa=3$ in~\eqref{eq:attack}
The achieved accuracy is shown in \autoref{fig:mnist-mal-het2}.
Additionally,
\autoref{fig:mnist-mal-het2-train} illustrates the behavior of two different algorithms during training,
where the solid lines show the average values and the grey filled area delimits one standard deviation interval across all agents at each iteration.
For the sake of computational time,
the global loss is computed at each iteration on a random subset (chosen beforehand and fixed through the iterations) of the global training set.
The influence of malicious agents causes the loss produced by~\eqref{eq:dgd} to steadily increase in the leftmost and center panels of \autoref{fig:mnist-mal-het2-train-dgd},
which quickly settles the accuracy to a suboptimal value (rightmost panel).
On the other hand,
setting just $\lambda=0.25$ in~\eqref{eq:fj-dgd-2} provides a good level of resilience as particularly evident by the global loss in the center panel of~\autoref{fig:mnist-mal-het2-train-fjdgd} that initially increases but recovers a decreasing trend after about $100$ iterations.
In this case,
the gap between accuracy achieved with our approach and~\eqref{eq:dgd} local training after $300$ iterations is even sharper,
as~\autoref{fig:mnist-mal-het2-300} highlights.

\paragraph*{Early stopping}
Contrary to the experiment without malicious agents in \autoref{fig:mnist-no-mal},
running~\eqref{eq:fj-dgd-2} with $\lambda\in\{0.25,0.5\}$ for $1000$ iterations causes several agents to overfit the training data with malicious agents.
To overcome this issue,
we implement an early stopping policy during training.
At every iteration,
each honest agent $i$ computes a moving average (MA) of the accuracy on its local training set over a sliding window of length $W$,
storing the maximal (so smoothed) accuracy achieved so far along with the corresponding local parameter $x_i$.
If the maximal MA of the accuracy does not increase,
\ie no improvement is done,
for more than $W_\text{imp}$ consecutive iterations,
the agent resets its local parameter to the one corresponding to the maximal accuracy MA and stops local training,
but keeps transmitting its local parameter to neighbors.
We set $W = W_\text{imp} = 20$ after trial-and-error.

\autoref{fig:mnist-mal-het2-1000} illustrates the outcome with early stopping.
The latter makes~\eqref{eq:dgd} slightly degrade performance,
possibly due to more randomness in the updates,
thus in \autoref{fig:mnist-mal-het2-1000} we report the accuracy~\eqref{eq:dgd} without early stopping for fairness.
Nonetheless,
our approach still outperforms both~\eqref{eq:dgd} and local training with $\lambda=0.25$,
some agents performing particularly well (global accuracy over $80\%$),
and achieves accuracy comparable to local training with $\lambda\in\{0.5,0.75\}$.

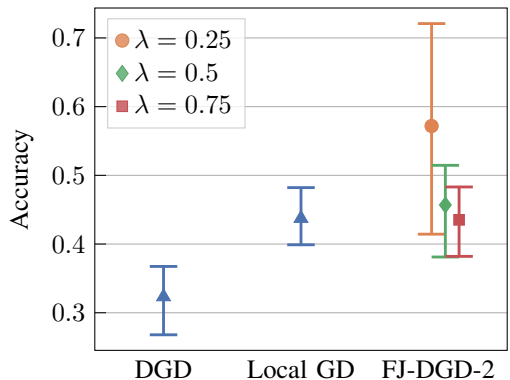
\begin{figure}
	\centering

\begin{tikzpicture}
	
	\definecolor{darkgray176}{RGB}{176,176,176}
	\definecolor{indianred1967882}{RGB}{196,78,82}
	\definecolor{lightgray204}{RGB}{204,204,204}
	\definecolor{mediumseagreen85168104}{RGB}{85,168,104}
	\definecolor{peru22113282}{RGB}{221,132,82}
	\definecolor{steelblue76114176}{RGB}{76,114,176}
	
	\begin{axis}[
		scale=.8,
		legend cell align={left},
		legend style={fill opacity=0.8, 
			draw opacity=1, 
			text opacity=1, 
			draw=lightgray204},
		legend pos=north west,
		tick align=inside,
		tick pos=left,
		unbounded coords=jump,
		x grid style={darkgray176},
		xmin=-0.5, xmax=2.5,
		xtick style={color=black},
		xtick={0,1,2},
		xticklabels={DGD,Local GD,FJ-DGD-2},
		xtick style={color=white},
		y grid style={darkgray176},
		ytick={.3,.4,.5,.6,.7},
		ylabel={Accuracy},
		ymajorgrids,
		ymin=0.245267468690872, ymax=0.743550103902817,
		ytick style={color=black}
		]
		\addplot [line width=1.08pt, steelblue76114176, forget plot]
		table {%
			0 0.267916679382324
			0 0.367409914731979
		};
		\addplot [line width=1.08pt, steelblue76114176, forget plot]
		table {%
			0.1 0.267916679382324
			-0.1 0.267916679382324
		};
		\addplot [line width=1.08pt, steelblue76114176, forget plot]
		table {%
			0.1 0.367409914731979
			-0.1 0.367409914731979
		};
		\addplot [line width=1.08pt, steelblue76114176, forget plot]
		table {%
			1 0.398930184543133
			1 0.482128366827965
		};
		\addplot [line width=1.08pt, steelblue76114176, forget plot]
		table {%
			1.1 0.398930184543133
			0.9 0.398930184543133
		};
		\addplot [line width=1.08pt, steelblue76114176, forget plot]
		table {%
			1.1 0.482128366827965
			0.9 0.482128366827965
		};
		\addplot [draw=steelblue76114176, fill=steelblue76114176, forget plot, mark=triangle*, only marks, mark size=3]
		table{%
			0 0.323246240615845
			1 0.437324345111847
		};
		\addplot [line width=1.08pt, peru22113282, forget plot]
		table {%
			1.95 0.571798801422119
		};
		\addplot [line width=1.08pt, peru22113282, forget plot]
		table {%
			1.95 0.414369378238916
			1.95 0.720900893211365
		};
		\addplot [line width=1.08pt, peru22113282, forget plot]
		table {%
			1.85 0.414369378238916
			2.05 0.414369378238916
		};
		\addplot [line width=1.08pt, peru22113282, forget plot]
		table {%
			1.85 0.720900893211365
			2.05 0.720900893211365
		};
		\addplot [draw=peru22113282, fill=peru22113282, mark=*, only marks, mark size=2.5]
		table{%
			1.95 0.571798801422119
		};
		\addlegendentry{$\lambda=0.25$}
		\addplot [line width=1.08pt, mediumseagreen85168104, forget plot]
		table {%
			2.05 0.456847846508026
		};
		\addplot [line width=1.08pt, mediumseagreen85168104, forget plot]
		table {%
			2.05 0.381159897893667
			2.05 0.51460587978363
		};
		\addplot [line width=1.08pt, mediumseagreen85168104, forget plot]
		table {%
			1.95 0.381159897893667
			2.15 0.381159897893667
		};
		\addplot [line width=1.08pt, mediumseagreen85168104, forget plot]
		table {%
			1.95 0.51460587978363
			2.15 0.51460587978363
		};
		\addplot [draw=mediumseagreen85168104, fill=mediumseagreen85168104, mark=diamond*, only marks, mark size=3]
		table{%
			2.05 0.456847846508026
		};
		\addlegendentry{$\lambda=0.5$}
		\addplot [line width=1.08pt, indianred1967882, forget plot]
		table {%
			2.15 0.435244292020798
		};
		\addplot [line width=1.08pt, indianred1967882, forget plot]
		table {%
			2.15 0.382060810923576
			2.15 0.483096841722727
		};
		\addplot [line width=1.08pt, indianred1967882, forget plot]
		table {%
			2.05 0.382060810923576
			2.25 0.382060810923576
		};
		\addplot [line width=1.08pt, indianred1967882, forget plot]
		table {%
			2.05 0.483096841722727
			2.25 0.483096841722727
		};
		\addplot [draw=indianred1967882, fill=indianred1967882, mark=square*, only marks, mark size=2]
		table{%
			2.15 0.435244292020798
		};
		\addlegendentry{$\lambda=0.75$}
	\end{axis}
	
\end{tikzpicture}
	\caption{Accuracy on MNIST ``Het-5'' in \cref{ex:task-mnist} with malicious agents after $1000$ learning iterations.}
	\label{fig:mnist-mal-het5}
\end{figure}

\subsubsection{Dataset ``Het-5''}
Finally,
we train on the dataset ``Heterogeneous-$5$'' for $1000$ iterations using early stopping.
\autoref{fig:mnist-mal-het5} reports the achieved accuracy which shows the same qualitative behavior observed with ``Heterogeneous-$2$.''
In this case,
the gap between our approach and the baselines is ever wider.
Setting $\lambda=0.25$ sizably outperforms both~\eqref{eq:dgd} and ``Local GD'',
while $\lambda\in\{0.5,0.75\}$ yields slightly higher accuracy than ``Local GD''.

\section{Limitations and Future Research}\label{sec:discussion}

This work provides formal guarantees and study the performance achieved with a globally assigned parameter $\lambda$.
More realistic scenarios might require that each agent $i$ locally sets its own parameter $\lambda_i$.
This raises both technical and practical questions,
such as if and how convergence can be characterized and how should the agents meaningfully set local parameters $\lambda_i$'s.
It may be particularly challenging with malicious agents,
whereby a poor choice of $\lambda_i$ may degrade resilience and yield low performance.
Literature on resilient consensus offers interesting options to dynamically adjust weights,
such as heuristic metrics of dissimilarity with neighbors~\cite{Baras19med-trust,Bonagura23acc-resilientConsensusEvidence} or inspired by the Hegselmann-Krause model~\cite{Hegselmann02jasss-boundedConfidence}.
These simple mechanisms may be effective to locally tune competition parameters $\lambda_i$'s in a decentralized manner during training.

Related to the previous discussion,
it is interesting to pair robust training with detection of adversaries.
For example,
the algorithms in~\cite{Yemini25tac-resilientDistributedOptim,Ballotta24acc-trustConfidence} use trust observations obtained from the wireless channel to filter out suspicious messages at each iteration,
and identify all adversaries in finite time almost surely.
More broadly,
if active and/or time-consuming security mechanisms are available,
the approach proposed in the present work can make the early training rounds resilient and be replaced by a standard distributed learning algorithm after adversaries have been identified.

\section{Conclusion}\label{sec:conclusion}

We have proposed a lightweight distributed learning algorithm which accommodates for personalization and resilience by combining distributed gradient descent and Friedkin-Johnsen model.
We quantified both its geometric convergence rate and the worst-case distance from the nominal global optimum in the presence of attacks,
and experimentally showed its effectiveness by testing it on various classification tasks.
In particular, we showcased its ability to learn personalized models that retain good generalization and demonstrated how it enhances resilience against attacks aimed at disrupting the training.
	

\begin{thebibliography}{10}
\providecommand{\url}[1]{#1}
\csname url@samestyle\endcsname
\providecommand{\newblock}{\relax}
\providecommand{\bibinfo}[2]{#2}
\providecommand{\BIBentrySTDinterwordspacing}{\spaceskip=0pt\relax}
\providecommand{\BIBentryALTinterwordstretchfactor}{4}
\providecommand{\BIBentryALTinterwordspacing}{\spaceskip=\fontdimen2\font plus
\BIBentryALTinterwordstretchfactor\fontdimen3\font minus
  \fontdimen4\font\relax}
\providecommand{\BIBforeignlanguage}[2]{{%
\expandafter\ifx\csname l@#1\endcsname\relax
\typeout{** WARNING: IEEEtran.bst: No hyphenation pattern has been}%
\typeout{** loaded for the language `#1'. Using the pattern for}%
\typeout{** the default language instead.}%
\else
\language=\csname l@#1\endcsname
\fi
#2}}
\providecommand{\BIBdecl}{\relax}
\BIBdecl

\bibitem{Yu22ral-DiNNO}
J.~Yu, J.~A. Vincent, and M.~Schwager, ``{{DiNNO}}: {{Distributed Neural
  Network Optimization}} for {{Multi-Robot Collaborative Learning}},''
  \emph{IEEE Robot. Autom. Lett.}, vol.~7, no.~2, pp. 1896--1903, 2022.

\bibitem{Ma21tvt-distributedLearningPlatoon}
X.~Ma, J.~Zhao, and Y.~Gong, ``Joint {{Scheduling}} and {{Resource Allocation}}
  for {{Efficiency-Oriented Distributed Learning Over Vehicle Platooning
  Networks}},'' \emph{IEEE Trans. Veh. Technol.}, vol.~70, no.~10, pp.
  10\,894--10\,908, 2021.

\bibitem{Huang18acm-distributedLearningSmartGrid}
H.~Huang, H.~Xu, Y.~Cai, R.~S. Khalid, and H.~Yu, ``Distributed {{Machine
  Learning}} on {{Smart-Gateway Network}} toward {{Real-Time Smart-Grid Energy
  Management}} with {{Behavior Cognition}},'' \emph{ACM Trans. Des. Autom.
  Electron. Syst.}, vol.~23, no.~5, pp. 56:1--56:26, 2018.

\bibitem{Sommer21ijf-distributedLearningWind}
B.~Sommer, P.~Pinson, J.~W. Messner, and D.~Obst, ``Online distributed learning
  in wind power forecasting,'' \emph{Int. J. Forecasting}, vol.~37, no.~1, pp.
  205--223, 2021.

\bibitem{Le24cst-distributedLearningIoT}
M.~Le, T.~{Huynh-The}, T.~{Do-Duy}, T.-H. Vu, W.-J. Hwang, and Q.-V. Pham,
  ``Applications of {{Distributed Machine Learning}} for the
  {{Internet-of-Things}}: {{A Comprehensive Survey}},'' \emph{IEEE Commun.
  Surv. Tuts.}, pp. 1--1, 2024.

\bibitem{Ma23pieee-TrustedAIinMultiagentSystems}
C.~Ma, J.~Li, K.~Wei, B.~Liu, M.~Ding, L.~Yuan, Z.~Han, and H.~Vincent~Poor,
  ``Trusted {{AI}} in {{Multiagent Systems}}: {{An Overview}} of {{Privacy}}
  and {{Security}} for {{Distributed Learning}},'' \emph{Proc. IEEE}, vol. 111,
  no.~9, pp. 1097--1132, 2023.

\bibitem{Nedic18pieee-distributedOptimizationSurvey}
A.~Nedi{\'c}, A.~Olshevsky, and M.~G. Rabbat, ``Network {{Topology}} and
  {{Communication-Computation Tradeoffs}} in {{Decentralized Optimization}},''
  \emph{Proc. IEEE}, vol. 106, no.~5, pp. 953--976, 2018.

\bibitem{Notarstefano19fdsc-distributedOptimizationSurvey}
G.~Notarstefano, I.~Notarnicola, and A.~Camisa, ``Distributed {{Optimization}}
  for {{Smart Cyber-Physical Networks}},'' \emph{Found. Trends Syst. Control.},
  vol.~7, no.~3, pp. 253--383, 2019.

\bibitem{Xin20spm-DecentralizedStochasticOptimization}
R.~Xin, S.~Kar, and U.~A. Khan, ``Decentralized {{Stochastic Optimization}} and
  {{Machine Learning}}: {{A Unified Variance-Reduction Framework}} for {{Robust
  Performance}} and {{Fast Convergence}},'' \emph{IEEE Signal Proc. Mag.},
  vol.~37, no.~3, pp. 102--113, 2020.

\bibitem{Chang20spm-DistributedLearningNonconvex}
T.-H. Chang, M.~Hong, H.-T. Wai, X.~Zhang, and S.~Lu, ``Distributed
  {{Learning}} in the {{Nonconvex World}}: {{From}} batch data to streaming and
  beyond,'' \emph{IEEE Signal Proc. Mag.}, vol.~37, no.~3, pp. 26--38, 2020.

\bibitem{Lyu24tnnls-PrivacyRobustnessFL}
L.~Lyu, H.~Yu, X.~Ma, C.~Chen, L.~Sun, J.~Zhao, Q.~Yang, and P.~S. Yu,
  ``Privacy and {{Robustness}} in {{Federated Learning}}: {{Attacks}} and
  {{Defenses}},'' \emph{IEEE Trans. Neural Netw. Learn. Syst.}, vol.~35, no.~7,
  pp. 8726--8746, 2024.

\bibitem{Liu24arxiv-distributedLearningSurvey}
C.~Liu, N.~Bastianello, W.~Huo, Y.~Shi, and K.~H. Johansson, ``A survey on
  secure decentralized optimization and learning,'' 2024.

\bibitem{Sundaram19tac-resilientDistributedOptimization}
S.~Sundaram and B.~Gharesifard, ``Distributed {{Optimization Under Adversarial
  Nodes}},'' \emph{IEEE Trans. Autom. Control}, vol.~64, no.~3, pp. 1063--1076,
  2019.

\bibitem{Fang22tsipn-BRIDGE}
C.~Fang, Z.~Yang, and W.~U. Bajwa, ``{{BRIDGE}}: {{Byzantine-Resilient
  Decentralized Gradient Descent}},'' \emph{IEEE Trans. Signal Inf. Process.
  Netw.}, vol.~8, pp. 610--626, 2022.

\bibitem{Shang22tcs-medianBasedResilientConsensus}
Y.~Shang, ``Median-{{Based Resilient Consensus Over Time-Varying Random
  Networks}},'' \emph{IEEE Trans. Circuits Syst. II}, vol.~69, no.~3, pp.
  1203--1207, 2022.

\bibitem{Li22tro-resilientDistributedLearningRobots}
J.~Li, W.~Abbas, M.~Shabbir, and X.~Koutsoukos, ``Byzantine {{Resilient
  Distributed Learning}} in {{Multirobot Systems}},'' \emph{IEEE Trans.
  Robot.}, vol.~38, no.~6, pp. 3550--3563, 2022.

\bibitem{Abbas18tcns-trustedNodes}
W.~Abbas, A.~Laszka, and X.~Koutsoukos, ``Improving {{Network Connectivity}}
  and {{Robustness Using Trusted Nodes With Application}} to {{Resilient
  Consensus}},'' \emph{IEEE Trans. Control Netw. Syst.}, vol.~5, no.~4, pp.
  2036--2048, 2018.

\bibitem{Yemini25tac-resilientDistributedOptim}
M.~Yemini, A.~Nedi{\'c}, A.~J. Goldsmith, and S.~Gil, ``Resilient {{Distributed
  Optimization}} for {{Multi-Agent Cyberphysical Systems}},'' \emph{IEEE Trans.
  Autom. Control}, 2025.

\bibitem{Ballotta24acc-trustConfidence}
L.~Ballotta and M.~Yemini, ``The {{Role}} of {{Confidence}} for {{Trust-Based
  Resilient Consensus}},'' in \emph{Proc. {{American Control Conf}}.}, 2024,
  pp. 2822--2829.

\bibitem{Hadjicostis22cdc-trustworthyConsensus}
C.~N. Hadjicostis and A.~D. {Dom{\'i}nguez-Garc{\'i}a}, ``Trustworthy
  {{Distributed Average Consensus}},'' in \emph{Proc. {{IEEE Conf}}. {{Decis}}.
  {{Contol}}}, 2022, pp. 7403--7408.

\bibitem{Fallah20nips-personalizedFL}
A.~Fallah, A.~Mokhtari, and A.~Ozdaglar, ``Personalized {{Federated Learning}}
  with {{Theoretical Guarantees}}: {{A Model-Agnostic Meta-Learning
  Approach}},'' in \emph{Proc {{NeurIPS}}}, vol.~33, 2020, pp. 3557--3568.

\bibitem{Dinh20nips-personalizedFL}
C.~T. Dinh, N.~Tran, and J.~Nguyen, ``Personalized {{Federated Learning}} with
  {{Moreau Envelopes}},'' in \emph{Proc. {{NeurIPS}}}, vol.~33, 2020, pp.
  21\,394--21\,405.

\bibitem{Hanzely22tmlr-personalizedFL}
F.~Hanzely, B.~Zhao, and M.~Kolar, ``Personalized {{Federated Learning}}: {{A
  Unified Framework}} and {{Universal Optimization Techniques}},'' \emph{Trans.
  Mach. Learn. Res.}, 2022.

\bibitem{Sadiev22jco-decentralizedPersonalizedFL}
A.~Sadiev, E.~Borodich, A.~Beznosikov, D.~Dvinskikh, S.~Chezhegov,
  R.~Tappenden, M.~Tak{\'a}{\v c}, and A.~Gasnikov, ``Decentralized
  personalized federated learning: {{Lower}} bounds and optimal algorithm for
  all personalization modes,'' \emph{EURO J. Comput. Optim.}, vol.~10, p.
  100041, 2022.

\bibitem{Toghani23lcss-PARSPush}
M.~T. Toghani, S.~Lee, and C.~A. Uribe, ``{{PARS-Push}}: {{Personalized}},
  {{Asynchronous}} and {{Robust Decentralized Optimization}},'' \emph{IEEE
  Control Syst. Lett.}, vol.~7, pp. 361--366, 2023.

\bibitem{Kundu22edge-robustnessPersonalizationFL}
A.~Kundu, P.~Yu, L.~Wynter, and S.~H. Lim, ``Robustness and {{Personalization}}
  in {{Federated Learning}}: {{A Unified Approach}} via {{Regularization}},''
  in \emph{{{IEEE Int}}. {{Conf}}. {{Edge Comp}}. {{Commun}}.}, 2022, pp.
  1--11.

\bibitem{Han23infocom-SplitGP}
D.-J. Han, D.-Y. Kim, M.~Choi, C.~G. Brinton, and J.~Moon, ``{{SplitGP}}:
  {{Achieving Both Generalization}} and {{Personalization}} in {{Federated
  Learning}},'' in \emph{Proc. {{IEEE Conf}}. {{Computer Commun}}.}, 2023, pp.
  1--10.

\bibitem{Bietti22icml-personalizationPrivacyAccuracy}
A.~Bietti, C.-Y. Wei, M.~Dudik, J.~Langford, and S.~Wu, ``Personalization
  {{Improves Privacy-Accuracy Tradeoffs}} in {{Federated Learning}},'' in
  \emph{Proc. {{Int}}. {{Conf}}. {{Mach}}. {{Learn}}.}\hskip 1em plus 0.5em
  minus 0.4em\relax PMLR, 2022, pp. 1945--1962.

\bibitem{Li21pmlr-Ditto}
T.~Li, S.~Hu, A.~Beirami, and V.~Smith, ``Ditto: {{Fair}} and {{Robust
  Federated Learning Through Personalization}},'' in \emph{Proc. {{Int}}.
  {{Conf}}. {{Mach}}. {{Learn}}.}\hskip 1em plus 0.5em minus 0.4em\relax PMLR,
  2021, pp. 6357--6368.

\bibitem{Ballotta24tac-competitionCollaboration}
L.~Ballotta, G.~Como, J.~S. Shamma, and L.~Schenato, ``Can {{Competition
  Outperform Collaboration}}? {{The Role}} of {{Misbehaving Agents}},''
  \emph{IEEE Trans. Autom. Control}, vol.~69, no.~4, pp. 2308--2323, 2024.

\bibitem{Bastianello21ecc-proximalGradientMethod}
N.~Bastianello and E.~Dall'Anese, ``Distributed and {{Inexact Proximal Gradient
  Method}} for {{Online Convex Optimization}},'' in \emph{Proc. {{European
  Control Conf}}.}, 2021, pp. 2432--2437.

\bibitem{Yuan16siamjo-convergenceDGD}
K.~Yuan, Q.~Ling, and W.~Yin, ``On the {{Convergence}} of {{Decentralized
  Gradient Descent}},'' \emph{SIAM J. Optim.}, vol.~26, no.~3, pp. 1835--1854,
  2016.

\bibitem{Yuan19tsp-ExactDiffusion}
K.~Yuan, B.~Ying, X.~Zhao, and A.~H. Sayed, ``Exact {{Diffusion}} for
  {{Distributed Optimization}} and {{Learning}}---{{Part I}}: {{Algorithm
  Development}},'' \emph{IEEE Trans. Signal Process.}, vol.~67, no.~3, pp.
  708--723, 2019.

\bibitem{Friedkin90jms-fjmodel}
N.~E. Friedkin and E.~C. Johnsen, ``Social influence and opinions,'' \emph{J.
  Math. Sociol.}, vol.~15, no. 3-4, pp. 193--206, 1990.

\bibitem{Proskurnikov17arc-tutorialModelingSocialNet}
A.~V. Proskurnikov and R.~Tempo, ``A tutorial on modeling and analysis of
  dynamic social networks. {{Part I}},'' \emph{Annu. Reviews Control}, vol.~43,
  pp. 65--79, 2017.

\bibitem{Li20mls-federatedOptimization}
T.~Li, A.~K. Sahu, M.~Zaheer, M.~Sanjabi, A.~Talwalkar, and V.~Smith,
  ``Federated {{Optimization}} in {{Heterogeneous Networks}},'' \emph{Proc.
  Mach. Learning Syst.}, vol.~2, pp. 429--450, 2020.

\bibitem{Baras19med-trust}
J.~S. Baras and X.~Liu, ``Trust is the {{Cure}} to {{Distributed Consensus}}
  with {{Adversaries}},'' in \emph{Proc. {{Mediterranean Conf}}. {{Control
  Autom}}.}, Akko, Israel, 2019, pp. 195--202.

\bibitem{Bonagura23acc-resilientConsensusEvidence}
V.~Bonagura, C.~Fioravanti, G.~Oliva, and S.~Panzieri, ``Resilient {{Consensus
  Based}} on {{Evidence Theory}} and {{Weight Correction}},'' in \emph{Proc.
  {{American Control Conf}}.}, 2023, pp. 393--398.

\bibitem{Hegselmann02jasss-boundedConfidence}
R.~Hegselmann and U.~Krause, ``Opinion {{Dynamics}} and {{Bounded Confidence}}:
  {{Models}}, {{Analysis}} and {{Simulation}},'' \emph{J. Artif. Soc. Soc.
  Simul.}, vol.~5, no.~3, 2002.

\end{thebibliography}
	

	\if0\mode

\begin{IEEEbiography}[{\includegraphics[width=1in,height=1.25in,clip,keepaspectratio]{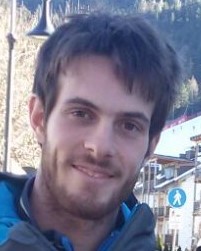}}]{Luca Ballotta}
	(Member, IEEE) is an Assistant Professor at the Department of Information Engineering, University of Padova, Padova, Italy.

He received the Ph.D. degree in information engineering from the University of Padova, Padova, Italy, in 2023.
He was a Postdoctoral Researcher at the Delft University of Technology, Delft, Netherlands in 2023-2025
and Visiting Student at the Massachusetts Institute of Technology in 2020 and 2022.
His research interests include resource allocation in multi-agent and network control systems,
resilient distributed control and learning,
and safe control.

Dr. Ballotta was the recipient of the Young Author Prize at the 2020 IFAC World Congress and was finalist at the 2024 EECI PhD Award.
\end{IEEEbiography}

\begin{IEEEbiography}[{\includegraphics[width=1in,height=1.25in,clip,keepaspectratio]{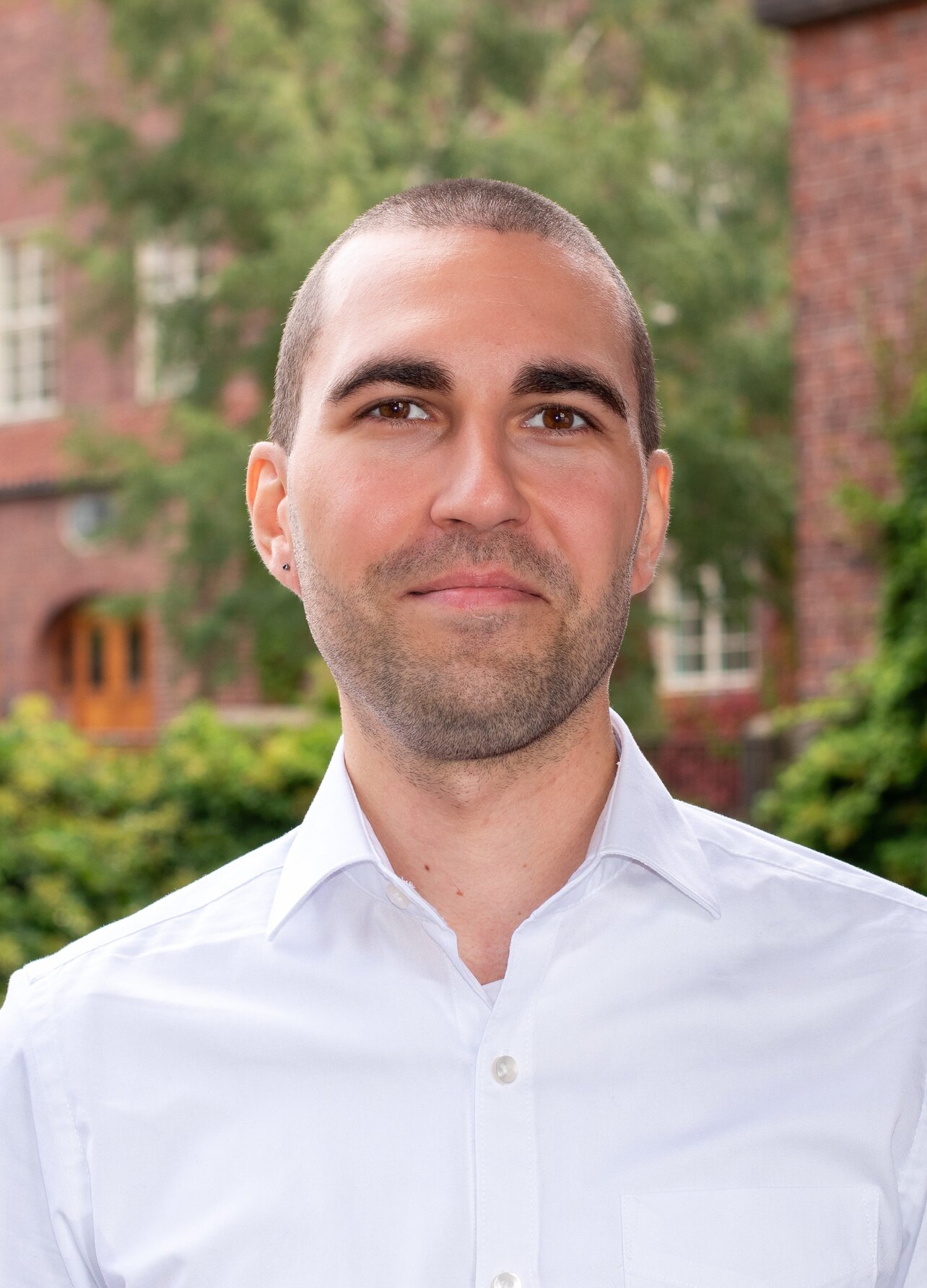}}]{Nicola Bastianello}
	(Member, IEEE) is a post-doc at the School of Electrical Engineering and Computer Science, KTH Royal Institute of Technology, Sweden. From 2021 to 2022 he was a post-doc at the Department of Information Engineering (DEI), University of Padova, Italy. He received the Ph.D. in Information Engineering at the University of Padova, Italy in 2021. During the Ph.D. he was a visiting student at the Department of Electrical, Computer, and Energy Engineering (ECEE), University of Colorado Boulder, Colorado, USA. He received the master degree in Automation Engineering (2018) and the bachelor degree in Information Engineering (2015) from the University of Padova, Italy. He currently serves in the IEEE CSS and EUCA Conference Editorial Boards. His research lies at the intersection of optimization and learning, with a focus on multi-agent systems.
\end{IEEEbiography}

\begin{IEEEbiography}[{\includegraphics[width=1in,height=1.25in,clip,keepaspectratio]{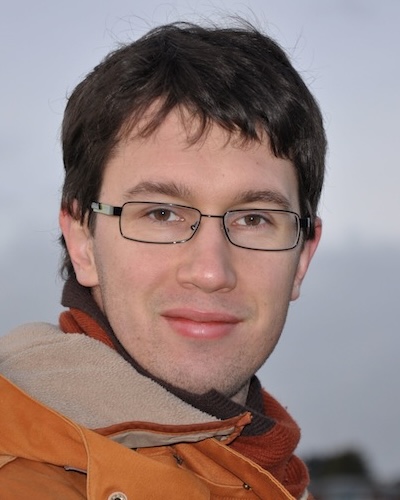}}]{Riccardo M. G. Ferrari }
	(Senior Member, IEEE) received the Laurea degree (cum laude and printing
honours) in electronic engineering and the Ph.D.
degree in information engineering from the University of Trieste, Italy, in 2004 and 2009, respectively.
He has held both academic and industrial research and
development positions, in particular as a Researcher
in the field of process instrumentation and control for
the steel-making sector. He is a Marie Curie alumnus
and currently an Associate Professor with Delft
Center for Systems and Control, Delft University
of Technology, The Netherlands. His research interests include wind power
fault tolerant control and fault diagnosis and attack detection in large-scale
cyber-physical systems, with applications to electric vehicles, cooperative
autonomous vehicles, and industrial control systems. He was a recipient of
the 2005 Giacomini Award of Italian Acoustic Society and he obtained the
2nd place in the Competition on Fault Detection and Fault Tolerant Control for
Wind Turbines during IFAC 2011. Furthermore, he was awarded an Honorable
Mention for the Pauk M. Frank Award at the IFAC SAFEPROCESS in
2018 and won an Airbus Award at IFAC 2020 for the best contribution to
the competition on Aerospace Industrial Fault Detection.
\end{IEEEbiography}

\begin{IEEEbiography}[{\includegraphics[width=1in,height=1.25in,clip,keepaspectratio]{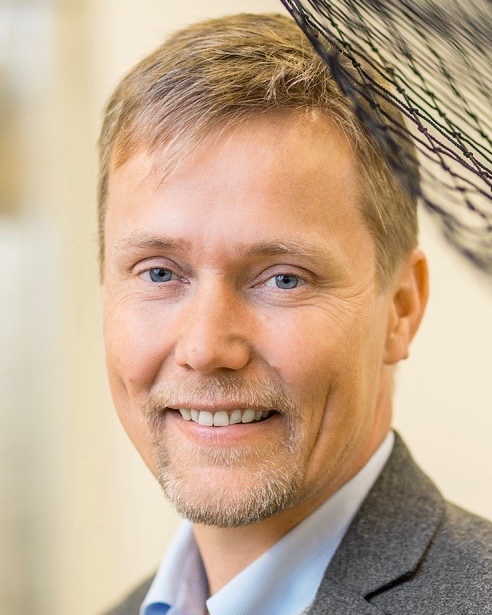}}]{Karl H. Johansson}
	(Fellow, IEEE) is Swedish Research Council Distinguished Professor in Electrical Engineering and Computer Science at KTH Royal Institute of Technology in Sweden and Founding Director of Digital Futures. He earned his MSc degree in Electrical Engineering and PhD in Automatic Control from Lund University. He has held visiting positions at UC Berkeley, Caltech, NTU and other prestigious institutions. His research interests focus on networked control systems and cyber-physical systems with applications in transportation, energy, and automation networks. For his scientific contributions, he has received numerous best paper awards and various distinctions from IEEE, IFAC, and other organizations. He has been awarded Distinguished Professor by the Swedish Research Council, Wallenberg Scholar by the Knut and Alice Wallenberg Foundation, Future Research Leader by the Swedish Foundation for Strategic Research. He has also received the triennial IFAC Young Author Prize and IEEE CSS Distinguished Lecturer. He is the recipient of the 2024 IEEE CSS Hendrik W. Bode Lecture Prize. His extensive service to the academic community includes being President of the European Control Association, IEEE CSS Vice President Diversity, Outreach \& Development, and Member of IEEE CSS Board of Governors and IFAC Council. He has served on the editorial boards of Automatica, IEEE TAC, IEEE TCNS and many other journals. He has also been a member of the Swedish Scientific Council for Natural Sciences and Engineering Sciences. He is Fellow of both the IEEE and the Royal Swedish Academy of Engineering Sciences.
\end{IEEEbiography}	
    \fi
	
\end{document}